\documentclass[twocolumn]{autart}
\usepackage{graphicx}
\usepackage[latin9]{inputenc}

\usepackage{amsmath}
\usepackage{amssymb}

\usepackage{array}
\usepackage{booktabs}
\usepackage{multirow}
\usepackage{makecell}
\usepackage{diagbox}
\usepackage{colortbl}

\usepackage{xcolor}
\usepackage{enumerate}
\usepackage{cite}
\usepackage{subfigure}
\usepackage{algorithm}
\usepackage{algpseudocode}

\usepackage[
    caption=false,
    font=normalsize,
    labelfont=sf,
    textfont=sf
]{subfig}
\usepackage{booktabs}
\usepackage{tabularx}
\usepackage{makecell}
\usepackage{array}

\newcolumntype{Y}{>{\centering\arraybackslash}X}

\input{tcilatex}

\algrenewcommand\algorithmicrequire{\textbf{Input:}}
\algrenewcommand\algorithmicensure{\textbf{Output:}}

\begin{document}

\begin{frontmatter}

\title{Data-Driven Policy Iteration Without an Initial
Stabilizing Policy: A Finite-Horizon Bootstrap
Method
}
\vspace{-4ex}

\thanks[*]{Corresponding author: Yang~Zhu.}

\author[a]{Jiacheng~Wu}\ead{jiachengwu@zju.edu.cn},
\author[a]{Yang~Zhu}*\ead{zhuyang88@zju.edu.cn}

\address[a]{College of Control Science and Engineering, Zhejiang University, Hangzhou 310027, China}
\vspace{-1ex}

\vspace{-1ex}
\begin{abstract}
This article investigates data-driven policy iteration (PI) for continuous-time linear systems without requiring an initially stabilizing policy. Standard infinite-horizon PI is not self-starting because its policy-evaluation step is well posed only when the feedback gain is stabilizing. However, verifying this property is difficult when the system matrices are unknown. To remove this requirement, we develop a finite-horizon bootstrap method. The key idea is to perform policy evaluation over a compact interval for a shifted system, where the evaluation equation is well defined for arbitrary bounded time-varying policies. We show that, for a sufficiently long horizon, the initial-time optimal gain of the shifted finite-horizon problem, when applied as a constant feedback gain, achieves a prescribed stability margin for the original system. We then derive a data-driven implementation from an off-policy identity evaluated along trajectories of the original plant. We use basis-function approximations to reconstruct the finite-horizon value matrix and policy, and we characterize the resulting error through a perturbed policy-improvement recursion. A data-driven Lyapunov certificate is further introduced to verify admissibility of the candidate gain before it is used to initialize infinite-horizon PI. Numerical studies of a batch reactor and a two-mass-spring system demonstrate the effectiveness of the proposed bootstrap method.
\end{abstract}
\end{frontmatter}

\section{Introduction}

Policy iteration (PI) is a widely used Newton-type method for continuous-time linear quadratic regulator (LQR) problems \cite{lewis2012optimal,hull2013optimal,liu2025adaptive}. PI is initialized with a stabilizing feedback gain and proceeds by repeatedly evaluating the current policy and updating the control law. The iteration preserves closed-loop stability, generates a monotonically nonincreasing sequence of value matrices, and converges locally at the Newton rate \cite{kleinman1968iterative,lopez2023efficient}. Integral and off-policy reinforcement learning (RL) formulations implement both steps using measured state--input data without explicitly identifying the system matrices \cite{cui2025learning,vrabie2009adaptive,jiang2012computational,shen2024data,luo2026recent,rizvi2018output}. Together, these properties make PI attractive for data-driven optimal control. Yet standard infinite-horizon PI is not self-starting \cite{gao2016adaptive,lian2022data,jiang2025off}. The policy-evaluation step in infinite-horizon PI requires the current feedback gain to stabilize the closed-loop system. If the gain is not stabilizing, closed-loop trajectories need not decay, the infinite-horizon cost is generally not finite, and the associated Lyapunov equation no longer defines a valid value matrix. Consequently, the set of admissible policies defines the domain on which infinite-horizon policy evaluation admits a valid value-function interpretation. When the system matrices are unknown, constructing and certifying an admissible gain may require precisely the model information that data-driven PI seeks to avoid. This circular dependence poses a structural obstacle to implementing infinite-horizon PI.

Value iteration (VI) avoids the need for an initial stabilizing feedback gain by initializing a value matrix instead \cite{possieri2022value,wei2016value,luo2019balancing}. In this setting, existing continuous-time VI schemes typically rely on diminishing or small step sizes and may use bounded-set projections or reset rules \cite{bian2016value,pang2020adaptive}. Their first-order updates may also converge more slowly than the local Newton iteration underlying PI. Hybrid iteration (HI) uses VI to reach the admissible-policy region and then switches to PI to exploit its fast local convergence \cite{gao2022resilient,wu2025distributed,qasem2023experimental}. The HI framework has also been applied to Markov jump systems \cite{shen2024secure,wang2026secure}, networked control systems \cite{zhang2025prescribed}, and autonomous vehicles \cite{liang2025cooperative}. Although HI reduces reliance on an initial stabilizing feedback gain, its VI stage still requires step-size selection, boundedness verification, and an explicit VI-to-PI switching criterion.

\begin{table*}[t]
\caption{Comparison of initialization requirements for continuous-time
data-driven RL algorithms.}
\label{tab:comparison_initialization}
\centering
\scriptsize
\setlength{\tabcolsep}{3.2pt}
\renewcommand{\arraystretch}{1.08}
\newcommand{\headercell}[1]{%
  \parbox[b][2.4\baselineskip][b]{\linewidth}{%
    \centering\bfseries #1}}
\newcommand{\headercellleft}[1]{%
  \parbox[b][2.4\baselineskip][b]{\linewidth}{%
    \raggedright\bfseries #1}}

\begin{tabularx}{\textwidth}{@{}>{\raggedright\arraybackslash}p{0.17\textwidth}YYYYY@{}}
\toprule
\headercellleft{References}
& \headercell{Initial stabilizing\\policy}
& \headercell{Stabilizing\\intermediate policy}
& \headercell{Infinite-horizon\\subproblem}
& \headercell{Step-size\\mechanism}
& \headercell{Bias\\mechanism} \\
\midrule

Standard PI \cite{jiang2012computational}
& Required
& Required
& Required
& N/A
& N/A \\[0.8pt]

VI \cite{bian2016value}
& Not required
& N/A
& Not required
& \makecell{Required}
& N/A \\[0.8pt]

HI \cite{gao2022resilient,wu2025distributed,qasem2023experimental}
& Not required
& \makecell{Required after\\switching}
& \makecell{Required after\\switching}
& Required
& N/A \\[0.8pt]

Bias PI \cite{jiang2022bias}
& Relaxed
& \makecell{Required for\\auxiliary system}
& Required
& N/A
& Required \\[0.8pt]

Homotopy PI \cite{chen2022homotopic,fan2025homotopy,ma2025adaptive}
& Not required
& Required
& Required
& N/A
& Required \\

\midrule
\rowcolor{black!5}
\textbf{This paper}
& \textbf{Not required}
& \textbf{Not required}
& \textbf{Not required}
& \textbf{Not required}
& \textbf{Not required} \\
\bottomrule
\end{tabularx}
\end{table*}

A second class of methods modifies the policy-evaluation operator to enlarge the initialization region. Bias PI introduces a bias parameter into policy evaluation and transfers the stability requirement from the original closed-loop matrix to a shifted auxiliary matrix \cite{jiang2022bias}. A smaller bias yields PI-like behavior, whereas a larger bias can expand the initialization region at the cost of making the iteration more VI-like. Related ideas have been explored for discrete-time systems \cite{yang2021model} and zero-sum games \cite{shen2025data,zhao2024novel}. Bias PI thus relaxes the initial admissibility requirement while retaining linear matrix equations at each iteration. Implementing the method requires selecting a sufficiently large bias to ensure that auxiliary policy evaluation is well posed. Global variants additionally use boundedness tests and bias-adjustment rules. 

Homotopy-based methods take a different approach. They start from an artificially stabilized system that admits a trivial or readily constructed stabilizing policy. They then gradually deform the auxiliary system into the original plant while maintaining a stabilizing policy for each intermediate problem. Homotopic PI methods move the artificial dynamics toward those of the original plant using a cumulative continuation factor and provide an explicit bound on the closed-loop pole location \cite{chen2022homotopic,ma2025adaptive}. A more recent predictor-corrector homotopy method constructs a path of stabilizing Riccati solutions and uses first-order differential information to predict the controller for the next auxiliary problem, thereby enabling faster continuation along the path \cite{fan2025homotopy}. Related homotopy-based methods have subsequently been extended to nonlinear systems \cite{chen2023adaptive}, inverse RL \cite{wu2025memory}, multiagent systems \cite{li2025cooperative}, and Markov jump systems \cite{wang2025parallel}. These methods substantially relax the requirement for an initially stabilizing policy for the original plant. Nevertheless, each continuation stage still involves an infinite-horizon subproblem, and the corresponding intermediate policy must be stabilizing to ensure that policy evaluation is well posed. Thus, homotopy methods address the initialization difficulty by constructing a stability-preserving path through a sequence of auxiliary infinite-horizon problems. Table 1 summarizes the initialization requirements of the representative continuous-time data-driven RL methods discussed above and highlights the distinctions between these methods and the proposed finite-horizon bootstrap method.

The difficulty of evaluating a nonstabilizing policy stems from the unbounded evaluation interval. By contrast, over a compact interval, the terminal-value policy-evaluation equation is well posed for any bounded time-varying feedback policy, regardless of closed-loop stability. This distinction suggests using finite-horizon PI as a bootstrap stage rather than as a replacement for the original infinite-horizon problem. The finite-horizon stage generates the stabilizing policy needed to initialize infinite-horizon PI. However, implementing this bootstrap from data raises two main challenges. First, the time-varying finite-horizon value matrix and policy must be reconstructed using trajectories of the original plant alone, without collecting data from the auxiliary shifted system. Second, approximation and regression errors inevitably affect this reconstruction. The resulting constant gain must therefore be independently certified to stabilize the original plant before being used to initialize infinite-horizon PI. The main contributions are as follows:

\begin{itemize}
\item We introduce a finite-horizon bootstrap mechanism that eliminates the need for an initially stabilizing policy. Unlike initialization methods based on bias PI \cite{jiang2022bias}, homotopy \cite{chen2022homotopic,ma2025adaptive}, or VI \cite{gao2022resilient,wu2025distributed,qasem2023experimental}, our method allows arbitrary bounded intermediate policies because it performs policy evaluation over a compact time interval.

\item We establish a finite-to-infinite-horizon connection showing that the initial-time gain generated by the finite-horizon problem approaches the optimal gain of a shifted infinite-horizon LQR problem as the horizon increases. This connection provides a principled route to a prescribed feasible stability margin.

\item We develop a data-driven off-policy implementation based solely on trajectories of the original plant and characterize the effects of basis-approximation and regression errors through a perturbed finite-horizon policy-improvement recursion. We further introduce an independent data-driven Lyapunov certificate to verify the admissibility of the candidate gain before using the gain to initialize infinite-horizon PI.
\end{itemize}

\noindent \textbf{Notation:} Throughout this article, we use $\mathbb{R}^{n}$ for the $n$-dimensional real Euclidean
space and $\mathbb{R}^{m\times n}$ for the space of real
$m\times n$ matrices. The identity matrix of order $n$ is
written as $I_n$, whereas $\mathbf{0}$ represents a zero scalar,
vector, or matrix whose dimensions are inferred from context. For a matrix $\textbf{A}$, we use $\sigma_{\min}(\textbf{A})$
to represent its smallest singular value, while
$\textbf{A}^{\dagger}$ refers to its Moore--Penrose
pseudoinverse. When $\textbf{A}$ is square,
$\sigma(\textbf{A})$ represents its spectrum, and its spectral
abscissa is given by
$
   \alpha(\textbf{A})
   =
   \max_{\lambda\in\sigma(\textbf{A})}
   \operatorname{Re}(\lambda)
$,
where $\operatorname{Re}(\lambda)$ represents the real part of
$\lambda$. Moreover,
$
    \mathbb{C}_{-}
    =
    \{z\in\mathbb{C}:\operatorname{Re}(z)<0\}
$
denotes the open left-half complex plane. The minimum and maximum eigenvalues of a symmetric matrix
$\textbf{X}$ are denoted by $\lambda_{\min}(\textbf{X})$ and
$\lambda_{\max}(\textbf{X})$, respectively. For matrix-valued
function $\textbf{H}$ defined on $[0,T]$, let
$
    \|\textbf{H}\|_{\infty,T}
    =
    \sup_{t\in[0,T]}\|\textbf{H}(t)\|_F .
$ Given a matrix $\textbf{A}=[a_{ij}]\in\mathbb{R}^{m\times n}$,
$\operatorname{vec}(\textbf{A})$ is formed by stacking the
columns of $\textbf{A}$ into a single vector.
For a symmetric matrix $\textbf{P}=[p_{ij}]\in\mathbb{R}^{n\times n}$, define
$\operatorname{vecs}(\textbf{P})=[p_{11},2p_{12},\ldots,2p_{1n},p_{22},2p_{23},\ldots,p_{nn}]^T.$ For a vector $x=[x_1,\ldots,x_n]^T\in\mathbb{R}^{n}$, define $\operatorname{vecv}(x)=[x_1^2,x_1x_2,\ldots,x_1x_n,x_2^2,x_2x_3,\ldots,x_n^2]^T.$ For matrices $X_1,\ldots,X_q$ having the same number of columns,
$\operatorname{col}(X_1,\ldots,X_q)
= [X_1^{T},\ldots,X_q^{T}]^{T}$
denotes their vertical concatenation.

\section{Preliminaries and problem statement}

Consider the continuous-time linear system
\begingroup
\setlength{\abovedisplayskip}{6pt}
\setlength{\belowdisplayskip}{6pt}
\setlength{\abovedisplayshortskip}{6pt}
\setlength{\belowdisplayshortskip}{6pt}
\begin{equation}
\dot{x}(t)
=\mathbf{A}x(t)+\mathbf{B}u(t),
\label{1}
\end{equation}
\endgroup
where $x(t)\in \mathbb{R}^{n}$ is the system state and
$u(t)\in \mathbb{R}^{m}$ is the control signal. The system
matrices $\mathbf{A}\in\mathbb{R}^{n\times n}$ and
$\mathbf{B}\in\mathbb{R}^{n\times m}$ are constant but
unavailable. The infinite-horizon quadratic cost is
\begingroup
\setlength{\abovedisplayskip}{6pt}
\setlength{\belowdisplayskip}{6pt}
\setlength{\abovedisplayshortskip}{6pt}
\setlength{\belowdisplayshortskip}{6pt}
\begin{equation}
J(x_{0},u)
=
\int_{0}^{\infty}
\left[
x^{T}(t)\mathbf{Q}x(t)
+
u^{T}(t)\mathbf{R}u(t)
\right]
\mathrm{d}t,
\label{2}
\end{equation}
\endgroup
where $\mathbf{Q}=\mathbf{Q}^{T}\geq \mathbf{0}$ and $\mathbf{R}=\mathbf{R}%
^{T}>\mathbf{0}$ are the state and
control penalties, respectively. The pair $(\mathbf{A},\mathbf{B})$ is assumed to
be stabilizable, and $(\mathbf{A},\mathbf{Q}^{1/2})$ is assumed to be
detectable.

\begin{definition}
A constant feedback gain $\mathbf{K}\in\mathbb{R}^{m\times n} $ is called
admissible if $\mathbf{A}-\mathbf{B}\mathbf{K}$ is Hurwitz. For a prescribed
constant $\varepsilon\geq 0$, $\mathbf{K}$ is called $\varepsilon$%
-admissible if
\begingroup
\setlength{\abovedisplayskip}{6pt}
\setlength{\belowdisplayskip}{6pt}
\setlength{\abovedisplayshortskip}{6pt}
\setlength{\belowdisplayshortskip}{6pt}
\begin{equation}
\max_{\lambda\in\sigma(\mathbf{A}-\mathbf{B}\mathbf{K})}
\operatorname{Re}(\lambda)
<
-\varepsilon.
\label{3}
\end{equation}
\endgroup
\end{definition}

When $\mathbf{A}$ and $\mathbf{B}$ are available, the optimal control law is
\begingroup
\setlength{\abovedisplayskip}{6pt}
\setlength{\belowdisplayskip}{6pt}
\setlength{\abovedisplayshortskip}{6pt}
\setlength{\belowdisplayshortskip}{6pt}
\begin{equation}
u^{*}(t)
=
-\mathbf{K}^{*}x(t)
=
-\mathbf{R}^{-1}\mathbf{B}^{T}\mathbf{P}^{*}x(t),
\label{4}
\end{equation}
\endgroup
where $\mathbf{P}^{*}=(\mathbf{P}^{*})^{T}\geq\mathbf{0}$ is the unique stabilizing solution
of the algebraic Riccati equation
\begingroup
\setlength{\abovedisplayskip}{6pt}
\setlength{\belowdisplayskip}{6pt}
\setlength{\abovedisplayshortskip}{6pt}
\setlength{\belowdisplayshortskip}{6pt}
\begin{equation}
\mathbf{A}^{T}\mathbf{P}^{*}
+
\mathbf{P}^{*}\mathbf{A}
-
\mathbf{P}^{*}\mathbf{B}\mathbf{R}^{-1}\mathbf{B}^{T}\mathbf{P}^{*}
+
\mathbf{Q}
=
\mathbf{0}.
\label{5}
\end{equation}
\endgroup
The standard continuous-time PI for solving \eqref{5} is
Kleinman's iteration. The following standard result of Kleinman's iteration is recalled.

\begin{lemma}[\hspace{-0.05em}{\protect\cite{kleinman1968iterative}}]
Let $\mathbf{K}^{[0]}$ be any admissible gain matrix, i.e., $\sigma (\mathbf{%
A}-\mathbf{B}\mathbf{K}^{[0]})\subset \mathbb{C}_{-}$. For $j=0,1,2,\ldots $%
, let $\mathbf{P}^{[j]}=(\mathbf{P}^{[j]})^{T}$ solve
\begingroup
\setlength{\abovedisplayskip}{6pt}
\setlength{\belowdisplayskip}{6pt}
\setlength{\abovedisplayshortskip}{6pt}
\setlength{\belowdisplayshortskip}{6pt}
\begin{equation}
\begin{aligned}
&(\mathbf{A}-\mathbf{B}\mathbf{K}^{[j]})^{T}\mathbf{P}^{[j]}
+\mathbf{P}^{[j]}(\mathbf{A}-\mathbf{B}\mathbf{K}^{[j]}) \\
&\quad+\mathbf{Q}
+(\mathbf{K}^{[j]})^{T}\mathbf{R}\mathbf{K}^{[j]}
=\mathbf{0},
\end{aligned}
\label{6}
\end{equation}
\endgroup
and update the control gain by
\begingroup
\setlength{\abovedisplayskip}{6pt}
\setlength{\belowdisplayskip}{6pt}
\setlength{\abovedisplayshortskip}{6pt}
\setlength{\belowdisplayshortskip}{6pt}
\begin{equation}
\mathbf{K}^{[j+1]}
=
\mathbf{R}^{-1}\mathbf{B}^{T}\mathbf{P}^{[j]}.
\label{7}
\end{equation}
\endgroup
Then, the following properties hold:

\begin{enumerate}
\item $\sigma(\mathbf{A}-\mathbf{B}\mathbf{K}^{[j]})\subset\mathbb{C}_{-}$
for all $j\geq0$;

\item $\mathbf{P}^{*}\leq \mathbf{P}^{[j+1]}\leq\mathbf{P}^{[j]}$ for
all $j\geq0$;

\item $\lim_{j\rightarrow\infty}\mathbf{K}^{[j]}=\mathbf{K}^{*}$ and $%
\lim_{j\rightarrow\infty}\mathbf{P}^{[j]}=\mathbf{P}^{*}$.
\end{enumerate}
\end{lemma}

Lemma 1 shows that Kleinman's iteration is not self-starting, because its
well-posedness and stability-preserving property both require an admissible
initial gain. The problem addressed in this paper is
therefore to construct, from measured data, a constant gain $\mathbf{K}%
^{[0]} $ satisfying \eqref{3}, which can serve as the initializer for
infinite-horizon data-driven PI.

\subsection{Why a finite-horizon bootstrap method}

The central difficulty in removing the initial admissible controller is that
the infinite-horizon policy evaluation operator is intrinsically defined
only on stabilizing policies. For a fixed gain $\mathbf{K}$, the evaluation
step requires the Lyapunov equation
\begingroup
\setlength{\abovedisplayskip}{6pt}
\setlength{\belowdisplayskip}{6pt}
\setlength{\abovedisplayshortskip}{6pt}
\setlength{\belowdisplayshortskip}{6pt}
\begin{equation*}
\begin{aligned}
(\mathbf{A}-\mathbf{B}\mathbf{K})^{T}\mathbf{P}_{K}
+\mathbf{P}_{K}(\mathbf{A}-\mathbf{B}\mathbf{K})
+\mathbf{Q}
+\mathbf{K}^{T}\mathbf{R}\mathbf{K}
=\mathbf{0}.
\end{aligned}
\end{equation*}
\endgroup
When $\mathbf{A}-\mathbf{B}\mathbf{K}$ is Hurwitz, its solution admits the
value representation
\begingroup
\setlength{\abovedisplayskip}{6pt}
\setlength{\belowdisplayskip}{6pt}
\setlength{\abovedisplayshortskip}{6pt}
\setlength{\belowdisplayshortskip}{6pt}
\begin{equation*}
\mathbf{P}_{K}
=
\int_{0}^{\infty}
e^{(\mathbf{A}-\mathbf{B}\mathbf{K})^{T}t}
\left(\mathbf{Q}+\mathbf{K}^{T}\mathbf{R}\mathbf{K}\right)
e^{(\mathbf{A}-\mathbf{B}\mathbf{K})t}
\,\mathrm{d}t.
\end{equation*}
\endgroup
If $\mathbf{K}$ is not stabilizing, the state-transition matrix $e^{(\mathbf{%
A}-\mathbf{B}\mathbf{K})t}$ does not decay exponentially. Since the cost is
accumulated over the unbounded interval $[0,\infty)$, the integral is
generally not finite, and the evaluation no longer defines a valid cost
matrix.

As shown in Fig. 1, a finite-horizon evaluation avoids this obstruction because it is posed on a
compact interval. For any bounded time-varying gain $\mathbf{K}(t)$ on $%
[0,T] $, the terminal-value Lyapunov equation
\begingroup
\setlength{\abovedisplayskip}{6pt}
\setlength{\belowdisplayskip}{6pt}
\setlength{\abovedisplayshortskip}{6pt}
\setlength{\belowdisplayshortskip}{6pt}
\begin{equation}
\begin{aligned} -\dot{\mathbf P}(t) =&\,(\mathbf A-\mathbf B\mathbf
K(t))^T\mathbf P(t) +\mathbf P(t)(\mathbf A-\mathbf B\mathbf K(t)) \\
&+\mathbf Q+\mathbf K^T(t)\mathbf R\mathbf K(t), \qquad \mathbf P(T)=\mathbf
M \end{aligned}  \label{8}
\end{equation}
\endgroup
has a unique solution on $[0,T]$ with a terminal weight
$\textbf{M}=\textbf{M}^{T}\geq\textbf{0}$. This well-posedness does not require the
time-varying closed-loop system to be exponentially stable. Thus, unlike the
infinite-horizon evaluation, the finite-horizon evaluation is well defined
for arbitrary bounded policies and can therefore be used before an
admissible gain is available.

The finite-horizon stage must still produce a gain that is admissible for
the original infinite-horizon problem. To this end, the evaluation is
performed for the shifted dynamics $\dot x(t)=\mathbf{A}_{\rho }x(t)+\mathbf{%
B}u(t)$ with $\mathbf{A}_{\rho }=\mathbf{A}+\rho \mathbf{I}$ and $%
\rho>\varepsilon$. Let $\mathbf{K}_{\rho,T}^\ast(t)$ denote the optimal
finite-horizon gain of this shifted problem. As $T$ increases, the
initial-time gain $\mathbf{K}_{\rho,T}^\ast(0)$ approaches the stabilizing
infinite-horizon gain $\mathbf{K}_\rho^\ast=\mathbf{R}^{-1}\mathbf{B}^T%
\mathbf{P}_\rho^\ast , $ where $\mathbf{P}_\rho^\ast$ solves the shifted
algebraic Riccati equation. Since $\mathbf{A}_\rho-\mathbf{B}\mathbf{K}%
_\rho^\ast$ is Hurwitz, applying the same gain to the original system gives
\begingroup
\setlength{\abovedisplayskip}{6pt}
\setlength{\belowdisplayskip}{6pt}
\setlength{\abovedisplayshortskip}{6pt}
\setlength{\belowdisplayshortskip}{6pt}
\begin{equation*}
\alpha(\mathbf{A}-\mathbf{B}\mathbf{K}_\rho^\ast) = \alpha(\mathbf{A}_\rho-%
\mathbf{B}\mathbf{K}_\rho^\ast)-\rho <-\rho<-\varepsilon .
\end{equation*}
\endgroup
Therefore, by continuity of the spectral abscissa, $\mathbf{K}%
_{\rho,T}^\ast(0)$ is also $\varepsilon$-admissible for all sufficiently
large $T$.\begin{figure}[t]
\centering
\includegraphics[width=0.9\linewidth]{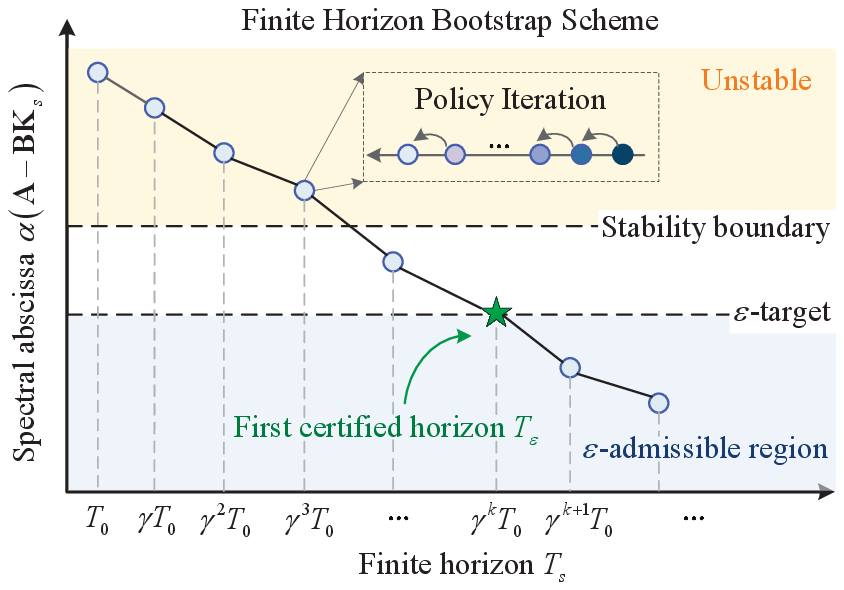}
\caption{Schematic of the finite-horizon bootstrap mechanism. At each horizon $T_s=\gamma^sT_0$, finite-horizon PI generates the candidate gain $\textbf{K}_s=\textbf{K}_{\rho,T_s}^{[\bar{j}_s]}(0)$. As the horizon is enlarged, the closed-loop spectral abscissa $\alpha_s=\alpha(\textbf{A}-\textbf{B}\textbf{K}_s)$ eventually enters the $\varepsilon$-admissible region. The first certified horizon is denoted by $T_\varepsilon$, and the inset shows the inner PI iterations at a fixed horizon.}
\end{figure}

\section{Model-based finite-horizon bootstrap algorithm}
This section presents the model-based version of the proposed finite-horizon
bootstrap method. The model-based construction clarifies the role of the
finite horizon and serves as the reference algorithm for the data-driven
development.

Let $\varepsilon \geq 0$ be the desired stability margin in \eqref{3}.
Choose a design parameter $\rho >\varepsilon $ and define the shifted system
matrix $\mathbf{A}_{\rho }$. The following assumption is imposed on the
shifted pair.

\textbf{Assumption 1.} The pair $(\mathbf{A}_{\rho},\mathbf{B})$ is
stabilizable, and $(\mathbf{A}_{\rho},\mathbf{Q}^{1/2})$ is detectable.

Assumption 1 guarantees the existence of the stabilizing solution $\mathbf{P}%
_{\rho }^{\ast }$ to the shifted algebraic Riccati equation
\begingroup
\setlength{\abovedisplayskip}{6pt}
\setlength{\belowdisplayskip}{6pt}
\setlength{\abovedisplayshortskip}{6pt}
\setlength{\belowdisplayshortskip}{6pt}
\begin{equation}
\mathbf{A}_{\rho }^{T}\mathbf{P}_{\rho }^{\ast }+\mathbf{P}_{\rho }^{\ast }%
\mathbf{A}_{\rho }-\mathbf{P}_{\rho }^{\ast }\mathbf{B}\mathbf{R}^{-1}%
\mathbf{B}^{T}\mathbf{P}_{\rho }^{\ast }+\mathbf{Q}=\mathbf{0}\text{.}
\label{9}
\end{equation}%
\endgroup
The control gain is $\mathbf{K}_{\rho }^{\ast }=\mathbf{R}^{-1}\mathbf{%
B}^{T}\mathbf{P}_{\rho }^{\ast }$. Since $\mathbf{A}_{\rho }-\mathbf{B}%
\mathbf{K}_{\rho }^{\ast }$ is Hurwitz, the matrix $\mathbf{A}-\mathbf{B}%
\mathbf{K}_{\rho }^{\ast }=\mathbf{A}_{\rho }-\mathbf{B}\mathbf{K}_{\rho
}^{\ast }-\rho \mathbf{I}$ is Hurwitz with a decay margin greater than $\rho
$. The purpose of the finite-horizon stage is to approximate this gain
without requiring an admissible initial policy.

For a horizon $T>0$ and a terminal matrix $\mathbf{M}=\mathbf{M}^{T}\geq
\mathbf{0}$, consider the shifted finite-horizon Riccati differential
equation
\begingroup
\setlength{\abovedisplayskip}{6pt}
\setlength{\belowdisplayskip}{6pt}
\setlength{\abovedisplayshortskip}{6pt}
\setlength{\belowdisplayshortskip}{6pt}
\begin{align}
-\dot{\mathbf{P}}_{\rho,T}^{\ast}(t)
={}&
\mathbf{A}_{\rho}^{T}\mathbf{P}_{\rho,T}^{\ast}(t)
+\mathbf{P}_{\rho,T}^{\ast}(t)\mathbf{A}_{\rho}
\notag\\
&-\mathbf{P}_{\rho,T}^{\ast}(t)
\mathbf{B}\mathbf{R}^{-1}\mathbf{B}^{T}
\mathbf{P}_{\rho,T}^{\ast}(t)
+\mathbf{Q}\text{,}
\label{10}
\end{align}
\endgroup
with terminal condition $\mathbf{P}_{\rho ,T}^{\ast }(T)=\mathbf{M}$. The
associated finite-horizon optimal gain is
\begingroup
\setlength{\abovedisplayskip}{6pt}
\setlength{\belowdisplayskip}{6pt}
\setlength{\abovedisplayshortskip}{6pt}
\setlength{\belowdisplayshortskip}{6pt}
\begin{equation}
\mathbf{K}_{\rho ,T}^{\ast }(t)=\mathbf{R}^{-1}\mathbf{B}^{T}\mathbf{P}%
_{\rho ,T}^{\ast }(t)\text{, } \qquad t\in \lbrack 0,T].  \label{11}
\end{equation}
\endgroup
The finite-horizon gain in \eqref{11} can be computed by a finite-horizon PI
without requiring a stabilizing initial policy. Let $\mathbf{K}_{\rho
,T}^{[0]}(t)$ be any bounded time-varying matrix on $[0,T]$. For $%
j=0,1,2,\ldots $, solve
\begingroup
\setlength{\abovedisplayskip}{6pt}
\setlength{\belowdisplayskip}{6pt}
\setlength{\abovedisplayshortskip}{6pt}
\setlength{\belowdisplayshortskip}{6pt}
\begin{align}
-\dot{\mathbf{P}}_{\rho ,T}^{[j]}(t)&=(\mathbf{A}_{\rho }-%
\mathbf{B}\mathbf{K}_{\rho ,T}^{[j]}(t))^{T}\mathbf{P}_{\rho ,T}^{[j]}(t)+%
\mathbf{P}_{\rho ,T}^{[j]}(t)(\mathbf{A}_{\rho }  \notag \\
&\!\!\!\!\!\!-\mathbf{B}\mathbf{K}_{\rho ,T}^{[j]}(t))+\mathbf{Q%
}+(\mathbf{K}_{\rho ,T}^{[j]}(t))^{T}\mathbf{R}\mathbf{K}_{\rho ,T}^{[j]}(t)
\label{12}
\end{align}%
\endgroup
with $\mathbf{P}_{\rho ,T}^{[j]}(T)=\mathbf{M}$, and update
\begingroup
\setlength{\abovedisplayskip}{6pt}
\setlength{\belowdisplayskip}{6pt}
\setlength{\abovedisplayshortskip}{6pt}
\setlength{\belowdisplayshortskip}{6pt}
\begin{equation}
\mathbf{K}_{\rho ,T}^{[j+1]}(t)=\mathbf{R}^{-1}\mathbf{B}^{T}\mathbf{P}%
_{\rho ,T}^{[j]}(t).  \label{13}
\end{equation}%
\endgroup
Because \eqref{12} is solved over a compact interval, the initial policy $%
\mathbf{K}_{\rho ,T}^{[0]}(t)$ need not stabilize the shifted system. As
will be shown below, the finite-horizon PI sequence converges to $\mathbf{K}%
_{\rho ,T}^{\ast }(t)$. After convergence, with $\bar{j}$ denoting the
terminal iteration index, the initializer is taken as $\mathbf{K}^{[0]}=%
\mathbf{K}_{\rho ,T}^{[\bar{j}]}(0)$. In the model-based setting,
admissibility can be checked directly by verifying
\begingroup
\setlength{\abovedisplayskip}{6pt}
\setlength{\belowdisplayskip}{6pt}
\setlength{\abovedisplayshortskip}{6pt}
\setlength{\belowdisplayshortskip}{6pt}
\begin{equation}
\max_{\lambda \in \sigma (\mathbf{A}-\mathbf{B}\mathbf{K}^{[0]})}\func{Re}%
(\lambda )<-\varepsilon .  \label{14}
\end{equation}%
\endgroup
If \eqref{14} holds, $\mathbf{K}^{[0]}$ is used to initialize Kleinman's
infinite-horizon PI \eqref{6}--\eqref{7} for the system \eqref{1}. Otherwise,
after the finite-horizon PI has converged for the current horizon, the
horizon is updated according to
\begingroup
\setlength{\abovedisplayskip}{6pt}
\setlength{\belowdisplayskip}{6pt}
\setlength{\abovedisplayshortskip}{6pt}
\setlength{\belowdisplayshortskip}{6pt}
\begin{equation}
T_{s+1}=\gamma T_{s},\qquad \gamma >1.  \label{15}
\end{equation}
\endgroup
\begin{algorithm}[t]
\caption{Model-Based Finite-Horizon Bootstrap Algorithm}
\renewcommand{\algorithmicrequire}{\textbf{Input:}}
\renewcommand{\algorithmicensure}{\textbf{Output:}}

\begin{algorithmic}[1]

\Require $\varepsilon \geq 0$, $\rho > \varepsilon$,
$\epsilon_{\mathrm{PI}} > 0$,
$T_0 > 0$, $\gamma > 1$, and
$\mathbf{M}=\mathbf{M}^{T}\geq\mathbf{0}$.

\Ensure An $\varepsilon$-admissible gain $\mathbf{K}^{[0]}$.

\State Set $s=0$ and $T_s=T_0$.

\Loop

    \State Choose a bounded policy
    $\mathbf{K}_{\rho,T_s}^{[0]}(t)$
    on $[0,T_s]$.

    \State Set $j=0$.

    \Loop

        \State Solve \eqref{12} on $[0,T_s]$.

        \State Update
        $\mathbf{K}_{\rho,T_s}^{[j+1]}(t)$
        according to \eqref{13}.

        \State Set $j \gets j+1$.

        \State Set
        $\mathbf{K}_{c}
        =
        \mathbf{K}_{\rho,T_s}^{[j]}(0)$.

        \If{
            $\displaystyle
            \max_{\lambda\in\sigma(
            \mathbf{A}-\mathbf{B}\mathbf{K}_{c})}
            \operatorname{Re}(\lambda)
            < -\varepsilon$
        }

            \State Set $\mathbf{K}^{[0]}=\mathbf{K}_{c}$.
            \State \textbf{return} $\mathbf{K}^{[0]}$.

        \EndIf

        \If{
            $\displaystyle
            \left\|
            \mathbf{K}_{\rho,T_s}^{[j]}
            -
            \mathbf{K}_{\rho,T_s}^{[j-1]}
            \right\|_{\infty,T_s}
            \leq \epsilon_{\mathrm{PI}}$
        }

            \State \textbf{break}.

        \EndIf

    \EndLoop

    \State Set
    $T_{s+1}=\gamma T_s$
    and $s \gets s+1$.

\EndLoop
\end{algorithmic}
\end{algorithm}
Algorithm 1 summarizes the model-based finite-horizon bootstrap procedure.
The following theorem justifies the algorithm by separating the
fixed-horizon finite-horizon PI property from the admissibility guarantee
obtained by enlarging the shifted horizon.

\begin{theorem}
Let $\rho >\varepsilon \geq 0$, $T>0$, and $\mathbf{M}=\mathbf{M}^{T}\geq
\mathbf{0}$. For any bounded piecewise-continuous initial policy $\mathbf{K}%
_{\rho ,T}^{[0]}(t)$ on $[0,T]$, generate $\mathbf{P}_{\rho ,T}^{[j]}(t)$
and $\mathbf{K}_{\rho ,T}^{[j+1]}(t)$ by \eqref{12}--\eqref{13}. Then, the
following statements hold.

\begin{enumerate}
\item For every $j\geq0$, $\mathbf{P}_{\rho,T}^{[j]}(t)$ is uniquely defined
on $[0,T]$.

\item The sequence satisfies
\begingroup
\setlength{\abovedisplayskip}{6pt}
\setlength{\belowdisplayskip}{6pt}
\setlength{\abovedisplayshortskip}{6pt}
\setlength{\belowdisplayshortskip}{6pt}
\begin{equation}
\mathbf{P}_{\rho ,T}^{\ast }(t)\leq \mathbf{P}_{\rho ,T}^{[j+1]}(t)\leq
\mathbf{P}_{\rho ,T}^{[j]}(t),\qquad \!\!\!\!t\in \lbrack 0,T],  \label{16}
\end{equation}%
\endgroup
and $\lim_{j\rightarrow \infty }\mathbf{P}_{\rho ,T}^{[j]}(t)=\mathbf{P}%
_{\rho ,T}^{\ast }(t),$ $\lim_{j\rightarrow \infty }\mathbf{K}_{\rho
,T}^{[j]}(t)=\mathbf{K}_{\rho ,T}^{\ast }(t)$ uniformly on $[0,T]$.

\item If Assumption~1 holds and $\mathbf{K}_{T}=\mathbf{K}_{\rho ,T}^{\ast
}(0)$, then $\lim_{T\rightarrow \infty }\mathbf{K}_{T}=\mathbf{K}_{\rho
}^{\ast }$. Moreover, there exists $T_{\varepsilon }>0$ such that, for all $%
T\geq T_{\varepsilon }$,
\begingroup
\setlength{\abovedisplayskip}{6pt}
\setlength{\belowdisplayskip}{6pt}
\setlength{\abovedisplayshortskip}{6pt}
\setlength{\belowdisplayshortskip}{6pt}
\begin{equation}
\alpha (\mathbf{A}-\mathbf{B}\mathbf{K}_{T})<-\varepsilon .  \label{17}
\end{equation}
\endgroup

\item If Assumption~1 holds, then for every $T\geq T_{\varepsilon }$, there
exists an integer $j_{\varepsilon }(T)$ such that
\begingroup
\setlength{\abovedisplayskip}{6pt}
\setlength{\belowdisplayskip}{6pt}
\setlength{\abovedisplayshortskip}{6pt}
\setlength{\belowdisplayshortskip}{6pt}
\begin{equation}
\alpha \!\left( \mathbf{A}-\mathbf{B}\mathbf{K}_{\rho ,T}^{[j]}(0)\right)
<-\varepsilon ,\qquad \forall j\geq j_{\varepsilon }(T).  \label{18}
\end{equation}
\endgroup
\end{enumerate}
\end{theorem}

\begin{proof}
(1) For each $j\geq 0$, define $\mathbf{A}^{[j]}(t)=\mathbf{A}_{\rho }-%
\mathbf{B}\mathbf{K}_{\rho ,T}^{[j]}(t)$ and $\mathbf{G}^{[j]}(t)=\mathbf{Q}%
+(\mathbf{K}_{\rho ,T}^{[j]}(t))^{T}\mathbf{R}\mathbf{K}_{\rho ,T}^{[j]}(t)$%
. Since $\mathbf{K}_{\rho ,T}^{[0]}$ is bounded and piecewise continuous on $%
[0,T]$, $\mathbf{A}^{[0]}(t)$ and $\mathbf{G}^{[0]}(t)$ are bounded on $%
[0,T] $. Therefore, the terminal-value Lyapunov equation \eqref{12} has a
unique symmetric solution on $[0,T]$. Let $\boldsymbol{\Phi }^{[j]}(\tau ,t)$
be the transition matrix generated by $\dot{z}(\tau )=\mathbf{A}^{[j]}(\tau
)z(\tau )$. The solution is
\begingroup
\setlength{\abovedisplayskip}{6pt}
\setlength{\belowdisplayskip}{6pt}
\setlength{\abovedisplayshortskip}{6pt}
\setlength{\belowdisplayshortskip}{6pt}
\begin{align}
\mathbf{P}_{\rho,T}^{[j]}(t)
={}&
\left(\boldsymbol{\Phi}^{[j]}(T,t)\right)^{T}
\mathbf{M}\boldsymbol{\Phi}^{[j]}(T,t)
\notag\\
&\!\!\!\!\!+
\int_{t}^{T}
\left(\boldsymbol{\Phi}^{[j]}(\tau,t)\right)^{T}
\mathbf{G}^{[j]}(\tau)
\boldsymbol{\Phi}^{[j]}(\tau,t)
\,\mathrm{d}\tau .
\label{19}
\end{align}
\endgroup
The right-hand side is finite for every $t\in \lbrack 0,T]$. Moreover, once $%
\mathbf{P}_{\rho ,T}^{[j]}$ is obtained, it is continuous on the compact
interval $[0,T]$ and hence bounded. Thus $\mathbf{K}_{\rho ,T}^{[j+1]}=%
\mathbf{R}^{-1}\mathbf{B}^{T}\mathbf{P}_{\rho ,T}^{[j]}$ is bounded. By
induction, the recursion is well defined for every $j\geq 0$.

(2) Consider the trajectory generated by the improved policy $u(t)=-\mathbf{K%
}_{\rho ,T}^{[j+1]}(t)x(t)$ for the shifted system. Along this trajectory,
\begingroup
\setlength{\abovedisplayskip}{6pt}
\setlength{\belowdisplayskip}{6pt}
\setlength{\abovedisplayshortskip}{6pt}
\setlength{\belowdisplayshortskip}{6pt}
\begin{align}
-\dot{\mathbf{P}}_{\rho,T}^{[j]}(t)
={}&
\bigl(\mathbf{A}^{[j+1]}(t)\bigr)^{T}
\mathbf{P}_{\rho,T}^{[j]}(t)
+\mathbf{P}_{\rho,T}^{[j]}(t)\mathbf{A}^{[j+1]}(t)
\notag\\
&+\mathbf{Q}
+\bigl(\mathbf{K}_{\rho,T}^{[j+1]}(t)\bigr)^{T}
\mathbf{R}\mathbf{K}_{\rho,T}^{[j+1]}(t)
\notag\\
&+\bigl(\Delta\mathbf{K}^{[j]}(t)\bigr)^{T}
\mathbf{R}\Delta\mathbf{K}^{[j]}(t).
\label{20}
\end{align}
\endgroup
where $\Delta \mathbf{K}^{[j]}(t)=\mathbf{K}_{\rho ,T}^{[j]}(t)-\mathbf{K}%
_{\rho ,T}^{[j+1]}(t)$. Subtracting the Lyapunov equation satisfied by $%
\mathbf{P}_{\rho ,T}^{[j+1]}$ gives
\begingroup
\setlength{\abovedisplayskip}{6pt}
\setlength{\belowdisplayskip}{6pt}
\setlength{\abovedisplayshortskip}{6pt}
\setlength{\belowdisplayshortskip}{6pt}
\begin{align}
-\dot{\mathbf{E}}^{[j]}(t)
={}&
\bigl(\mathbf{A}^{[j+1]}(t)\bigr)^{T}\mathbf{E}^{[j]}(t)
+\mathbf{E}^{[j]}(t)\mathbf{A}^{[j+1]}(t)
\notag\\
&\!\!\!\!\!\!\!\!+\bigl(\Delta\mathbf{K}^{[j]}(t)\bigr)^{T}
\mathbf{R}\Delta\mathbf{K}^{[j]}(t),
\qquad
\!\!\!\!\!\!\!\!\mathbf{E}^{[j]}(T)=\mathbf{0},
\label{21}
\end{align}
\endgroup
where $\mathbf{E}^{[j]}(t)=\mathbf{P}_{\rho ,T}^{[j]}(t)-\mathbf{P}_{\rho
,T}^{[j+1]}(t)$. Therefore,
\begingroup
\setlength{\abovedisplayskip}{6pt}
\setlength{\belowdisplayskip}{6pt}
\setlength{\abovedisplayshortskip}{6pt}
\setlength{\belowdisplayshortskip}{6pt}
\begin{align*}
\mathbf{E}^{[j]}(t)
={}&
\int_{t}^{T}
\left(\boldsymbol{\Phi}^{[j+1]}(\tau,t)\right)^{T}
\Delta\mathbf{K}^{[j]}(\tau)^{T} \\
&\times
\mathbf{R}\Delta\mathbf{K}^{[j]}(\tau)
\boldsymbol{\Phi}^{[j+1]}(\tau,t)
\,\mathrm{d}\tau
\geq \mathbf{0}.
\end{align*}
\endgroup
Thus $\mathbf{P}_{\rho ,T}^{[j+1]}(t)\leq \mathbf{P}_{\rho ,T}^{[j]}(t)$ for
$t\in \lbrack 0,T].$ Next, $\mathbf{P}_{\rho ,T}^{[j]}(t)$ is the
finite-horizon value matrix associated with the policy $u(t)=-\mathbf{K}%
_{\rho ,T}^{[j]}(t)x(t)$ for the shifted system. Since $\mathbf{P}_{\rho
,T}^{\ast }(t)$ is the optimal finite-horizon value matrix, we have $\mathbf{%
P}_{\rho ,T}^{\ast }(t)\leq \mathbf{P}_{\rho ,T}^{[j]}(t)$ for $t\in \lbrack
0,T]$. This proves the monotonicity relation. It remains to prove
convergence. Since $\mathbf{P}_{\rho ,T}^{\ast }(t)\leq \mathbf{P}_{\rho
,T}^{[j]}(t)\leq \mathbf{P}_{\rho ,T}^{[0]}(t)$, the sequence $\{\mathbf{P}%
_{\rho ,T}^{[j]}\}$ is uniformly bounded on $[0,T]$. Moreover, $\mathbf{K}%
_{\rho ,T}^{[j+1]}(t)=\mathbf{R}^{-1}\mathbf{B}^{T}\mathbf{P}_{\rho
,T}^{[j]}(t)$ implies that $\{\mathbf{K}_{\rho ,T}^{[j]}\}_{j\geq 1}$ is
also uniformly bounded on $[0,T]$. From \eqref{12}, the derivatives $\{\dot{%
\mathbf{P}}_{\rho ,T}^{[j]}\}_{j\geq 1}$ are uniformly bounded. Hence $\{%
\mathbf{P}_{\rho ,T}^{[j]}\}$ is equicontinuous on $[0,T]$. For each fixed $%
t\in \lbrack 0,T]$, the monotone bounded sequence $\{\mathbf{P}_{\rho
,T}^{[j]}(t)\}$ has a limit in the Loewner order $\mathbf{P}_{\infty }(t)$.
The equicontinuity on the compact interval $[0,T]$, together with the
pointwise convergence, implies that the convergence to $\mathbf{P}_{\infty }$
is uniform. Moreover, $\mathbf{K}_{\rho ,T}^{[j]}$ converges uniformly to $%
\mathbf{K}_{\infty }(t)=\mathbf{R}^{-1}\mathbf{B}^{T}\mathbf{P}_{\infty }(t)$%
. Passing to the limit in the integral form of \eqref{12} yields
\begingroup
\setlength{\abovedisplayskip}{6pt}
\setlength{\belowdisplayskip}{6pt}
\setlength{\abovedisplayshortskip}{6pt}
\setlength{\belowdisplayshortskip}{6pt}
\begin{align*}
-\dot{\mathbf{P}}_{\infty}(t)
={}&
\bigl(\mathbf{A}_{\rho}-\mathbf{B}\mathbf{K}_{\infty}(t)\bigr)^{T}
\mathbf{P}_{\infty}(t)
+\mathbf{P}_{\infty}(t) \\
&\times
\bigl(\mathbf{A}_{\rho}-\mathbf{B}\mathbf{K}_{\infty}(t)\bigr)
+\mathbf{Q}
+\mathbf{K}_{\infty}^{T}(t)\mathbf{R}\mathbf{K}_{\infty}(t),
\end{align*}
\endgroup
with $\mathbf{P}_{\infty }(T)=\mathbf{M}$, where $\mathbf{P}_{\infty }$
satisfies the finite-horizon Riccati equation \eqref{10}. By uniqueness of %
\eqref{10}, $\mathbf{P}_{\infty }(t)=\mathbf{P}_{\rho ,T}^{\ast }(t)$, $t\in
\lbrack 0,T]$. Therefore, $\mathbf{P}_{\rho ,T}^{[j]}$ and $\mathbf{K}_{\rho
,T}^{[j]}$ converge uniformly to $\mathbf{P}_{\rho ,T}^{\ast }$ and $\mathbf{%
K}_{\rho ,T}^{\ast }$, respectively, on $[0,T]$.

(3) Let $\mathbf{S}(\tau )$ be the solution of the forward Riccati flow
\begingroup
\setlength{\abovedisplayskip}{6pt}
\setlength{\belowdisplayskip}{6pt}
\setlength{\abovedisplayshortskip}{6pt}
\setlength{\belowdisplayshortskip}{6pt}
\begin{align}
\dot{\mathbf{S}}(\tau)
={}&
\mathbf{A}_{\rho}^{T}\mathbf{S}(\tau)
+\mathbf{S}(\tau)\mathbf{A}_{\rho}
\notag\\
&-\mathbf{S}(\tau)\mathbf{B}\mathbf{R}^{-1}
\mathbf{B}^{T}\mathbf{S}(\tau)
+\mathbf{Q},
\label{22}
\end{align}
\endgroup
with $\mathbf{S}(0)=\mathbf{M}$. By the time transformation $\tau =T-t$, the
finite-horizon Riccati solution satisfies $\mathbf{P}_{\rho ,T}^{\ast }(t)=%
\mathbf{S}(T-t)$, and in particular $\mathbf{P}_{\rho ,T}^{\ast }(0)=\mathbf{%
S}(T)$. Under Assumption~1, the continuous-time Riccati flow converges to
the stabilizing solution of the shifted algebraic Riccati equation \eqref{9}%
, that is, $\lim_{\tau \rightarrow \infty }\mathbf{S}(\tau )=\mathbf{P}%
_{\rho }^{\ast }.$ Therefore, $\lim_{T\rightarrow \infty }\mathbf{P}_{\rho
,T}^{\ast }(0)=\mathbf{P}_{\rho }^{\ast }$ and $\lim_{T\rightarrow \infty }%
\mathbf{K}_{\rho ,T}^{\ast }(0)=\mathbf{K}_{\rho }^{\ast }$. Since $\mathbf{A%
}_{\rho }-\mathbf{B}\mathbf{K}_{\rho }^{\ast }$ is Hurwitz,
\begingroup
\setlength{\abovedisplayskip}{6pt}
\setlength{\belowdisplayskip}{6pt}
\setlength{\abovedisplayshortskip}{6pt}
\setlength{\belowdisplayshortskip}{6pt}
\begin{equation*}
\alpha (\mathbf{A}-\mathbf{B}\mathbf{K}_{\rho }^{\ast })
=
\alpha (\mathbf{A}_{\rho }-\mathbf{B}\mathbf{K}_{\rho }^{\ast })
-\rho <-\rho <-\varepsilon .
\end{equation*}
\endgroup
Because the spectral abscissa is continuous with respect to the matrix
entries, and continuous with respect to $\mathbf{K}$, there exists a
neighborhood $\mathcal{N}$ of $\mathbf{K}_{\rho }^{\ast }$ such that $\alpha
(\mathbf{A}-\mathbf{B}\mathbf{K})<-\varepsilon $, $\forall \mathbf{K}\in
\mathcal{N}$. Since $\lim_{T\rightarrow \infty }\mathbf{K}_{\rho ,T}^{\ast
}(0)=\mathbf{K}_{\rho }^{\ast }$, there exists $T_{\varepsilon }>0$ such
that \eqref{17} holds.

(4) Fix any $T\geq T_{\varepsilon }$. From statement 3), $\alpha (\mathbf{A}-%
\mathbf{B}\mathbf{K}_{T})<-\varepsilon $. By continuity of the spectral
abscissa, there exists a neighborhood $\mathcal{N}_{T}$ of $\mathbf{K}_{T}$
such that $\alpha (\mathbf{A}-\mathbf{B}\mathbf{K})<-\varepsilon$, $\forall
\mathbf{K}\in \mathcal{N}_{T}$. From statement 2), $\lim_{j\rightarrow
\infty }\mathbf{K}_{\rho ,T}^{[j]}(0)=\mathbf{K}_{\rho ,T}^{\ast }(0)=%
\mathbf{K}_{T}$. Therefore, there exists an integer $j_{\varepsilon }(T)$
such that $\mathbf{K}_{\rho ,T}^{[j]}(0)\in \mathcal{N}_{T}$, $\forall j\geq
j_{\varepsilon }(T)$. This proves the theorem.
\end{proof}

The multiplicative update \eqref{15} gives $T_{s}=\gamma ^{s}T_{0}$. Since $%
\gamma >1$, $T_{s}\rightarrow \infty $ as $s\rightarrow \infty $. Theorem 2
implies that the horizon is sufficiently large if
\begingroup
\setlength{\abovedisplayskip}{6pt}
\setlength{\belowdisplayskip}{6pt}
\setlength{\abovedisplayshortskip}{6pt}
\setlength{\belowdisplayshortskip}{6pt}
\begin{equation}
s\geq s_{\varepsilon}
=
\max\left\{
0,
\left\lceil
\frac{\log(T_{\varepsilon}/T_{0})}{\log\gamma}
\right\rceil
\right\},
\label{23}
\end{equation}
\endgroup
then $T_{s}\geq T_{\varepsilon }$. Therefore, for every $s\geq
s_{\varepsilon }$, there exists $j_{\varepsilon }(T_{s})$ such that
\begingroup
\setlength{\abovedisplayskip}{6pt}
\setlength{\belowdisplayskip}{6pt}
\setlength{\abovedisplayshortskip}{6pt}
\setlength{\belowdisplayshortskip}{6pt}
\begin{equation}
\alpha \!\left(
\mathbf{A}-\mathbf{B}\mathbf{K}_{\rho,T_{s}}^{[j]}(0)
\right)
<-\varepsilon,
\qquad
\forall j\geq j_{\varepsilon}(T_{s}).
\label{24}
\end{equation}
\endgroup
Thus, the finite-horizon bootstrap mechanism ensures that an $\varepsilon$%
-admissible initializer can be obtained after a sufficiently large horizon
and sufficiently many finite-horizon PI iterations. In the model-based
algorithm, this admissibility is verified directly by the spectral test %
\eqref{14}.

\section{Data-driven finite-horizon bootstrap algorithm}

This section develops a data-driven realization of the finite-horizon PI stage in Algorithm 1. The data are collected from the original system $\dot{x}%
(t)=\mathbf{A}x(t)+\mathbf{B}u_{b}(t)$, where $u_{b}(t)$ is a behavior input.

For the current policy $u(t)=-\mathbf{K}_{\rho ,T}^{[j]}(t)x(t)$, define $%
\mathbf{Q}_{\rho ,T}^{[j]}(t)=\mathbf{Q}+(\mathbf{K}_{\rho ,T}^{[j]}(t))^{T}%
\mathbf{R}\mathbf{K}_{\rho ,T}^{[j]}(t)$ and $\mathbf{z}^{[j]}(t)=u_{b}(t)+%
\mathbf{K}_{\rho ,T}^{[j]}(t)x(t).$ The signal $\mathbf{z}^{[j]}(t)$ is the
off-policy correction term. Let $V^{[j]}(x,t)=x^{T}\mathbf{P}_{\rho
,T}^{[j]}(t)x$. Along the behavior trajectory, one obtains
\begingroup
\setlength{\abovedisplayskip}{6pt}
\setlength{\belowdisplayskip}{6pt}
\setlength{\abovedisplayshortskip}{6pt}
\setlength{\belowdisplayshortskip}{6pt}
\begin{align}
&\frac{\mathrm{d}}{\mathrm{d}t}
\left[
x^{T}(t)\mathbf{P}_{\rho,T}^{[j]}(t)x(t)
\right]
\notag\\
={}&
-x^{T}(t)\mathbf{Q}_{\rho,T}^{[j]}(t)x(t)
-2\rho x^{T}(t)\mathbf{P}_{\rho,T}^{[j]}(t)x(t)
\notag\\
&+2\left(
u_{b}(t)+\mathbf{K}_{\rho,T}^{[j]}(t)x(t)
\right)^{T}
\mathbf{R}\mathbf{K}_{\rho,T}^{[j+1]}(t)x(t).
\label{25}
\end{align}
\endgroup
Integrating \eqref{25} over each stored interval
$I_k=[t_k,t_{k+1}]\subset[0,T]$, $k=1,\ldots,N_d$, gives
\begingroup
\setlength{\abovedisplayskip}{6pt}
\setlength{\belowdisplayskip}{6pt}
\setlength{\abovedisplayshortskip}{6pt}
\setlength{\belowdisplayshortskip}{6pt}
\setlength{\jot}{2pt}
\begin{align}
&x^{T}(t_{k+1})\mathbf{P}_{\rho,T}^{[j]}(t_{k+1})x(t_{k+1})
-x^{T}(t_{k})\mathbf{P}_{\rho,T}^{[j]}(t_{k})x(t_{k})
\notag\\
={}&-\int_{t_{k}}^{t_{k+1}}
x^{T}(t)\mathbf{Q}_{\rho,T}^{[j]}(t)x(t)\,\mathrm{d}t
\notag\\
&+2\int_{t_{k}}^{t_{k+1}}
\bigl(u_{b}(t)+\mathbf{K}_{\rho,T}^{[j]}(t)x(t)\bigr)^{T}
\mathbf{R}\mathbf{K}_{\rho,T}^{[j+1]}(t)x(t)\,\mathrm{d}t
\notag\\
&-2\rho\int_{t_{k}}^{t_{k+1}}
x^{T}(t)\mathbf{P}_{\rho,T}^{[j]}(t)x(t)\,\mathrm{d}t.
\label{26}
\end{align}
\endgroup

\begin{remark}
Although the finite-horizon PI is formulated for the shifted matrix
$\mathbf{A}_\rho$, no trajectory of the shifted system is required. The
off-policy identity \eqref{26} is evaluated along the behavior trajectory of the
original plant \eqref{1}, and the effect of the shift is represented explicitly by the
term involving $\rho$. Therefore, when the horizon is enlarged from $T_s$ to
$T_{s+1}$, one only needs behavior data of the original system over the
enlarged interval.
\end{remark}
Choose differentiable basis functions $\phi _{\ell }(t)$, $\ell =1,\ldots
,N_{p}$ with $\phi _{\ell }(T)=0$ and basis functions $\psi _{\ell }(t)$, $%
\ell =1,\ldots ,N_{k}$. Approximate $\mathbf{P}_{\rho,T}^{[j]}(t)$ and $\mathbf{K}_{\rho,T}^{[j+1]}(t)$ as
\begingroup
\setlength{\abovedisplayskip}{6pt}
\setlength{\belowdisplayskip}{6pt}
\setlength{\abovedisplayshortskip}{6pt}
\setlength{\belowdisplayshortskip}{6pt}
\begin{equation}
\mathbf{P}_{\rho,T}^{[j]}(t)
=
\mathbf{M}
+\sum_{\ell=1}^{N_{p}}\phi_{\ell}(t)\mathbf{P}_{\ell}^{[j]}
+\Delta_{P}^{[j]}(t),
\label{27}
\end{equation}
\endgroup
where $\mathbf{P}_{\ell }^{[j]}=(\mathbf{P}_{\ell }^{[j]})^{T}$, and
\begingroup
\setlength{\abovedisplayskip}{6pt}
\setlength{\belowdisplayskip}{6pt}
\setlength{\abovedisplayshortskip}{6pt}
\setlength{\belowdisplayshortskip}{6pt}
\begin{equation}
\mathbf{K}_{\rho,T}^{[j+1]}(t)
=
\sum_{\ell=1}^{N_{k}}\psi_{\ell}(t)\mathbf{K}_{\ell}^{[j+1]}
+\Delta_{K}^{[j+1]}(t),
\label{28}
\end{equation}
\endgroup
where $\Delta _{P}^{[j]}(t)$ and $\Delta _{K}^{[j+1]}(t)$ denote the
approximation residuals induced by the chosen basis functions.
\begin{remark}
The expansions \eqref{27} and \eqref{28} are introduced to turn the
time-varying unknowns $\mathbf{P}_{\rho,T}^{[j]}(t)$ and $\mathbf{K}_{\rho,T}^{[j+1]}(t)$ in \eqref{26} into a finite set of coefficient matrices,
so that the off-policy identity  \eqref{26} can be solved from data as the linear
regression. In \eqref{27}, the term $\mathbf{M}$ and the condition
$\phi_\ell(T)=0$ enforce the terminal constraint
$\mathbf{P}_{\rho ,T}^{[j]}(T)=\mathbf{M}$ by construction. 
\end{remark}

To facilitate the solution of \eqref{26} using the representations in \eqref{27}-\eqref{28}, we define the following interval-wise quantities for each stored interval $I_{k}$
\begingroup
\setlength{\abovedisplayskip}{6pt}
\setlength{\belowdisplayskip}{6pt}
\setlength{\abovedisplayshortskip}{6pt}
\setlength{\belowdisplayshortskip}{6pt}
\setlength{\jot}{2pt}
\begin{align*}
\boldsymbol{\omega}_{P,k,\ell}
={}&
\phi_{\ell}(t_{k+1})
\mathrm{vecv}^{T}\!\bigl(x(t_{k+1})\bigr)
-\phi_{\ell}(t_{k})
\mathrm{vecv}^{T}\!\bigl(x(t_{k})\bigr)
\\
&+2\rho\int_{t_{k}}^{t_{k+1}}
\phi_{\ell}(t)\mathrm{vecv}^{T}\!\bigl(x(t)\bigr)
\,\mathrm{d}t,
\\[1mm]
\boldsymbol{\omega}_{K,k,\ell}^{[j]}
={}&
-2\int_{t_{k}}^{t_{k+1}}
\psi_{\ell}(t)
\left[
x^{T}(t)\otimes
\bigl(\mathbf{z}^{[j]}(t)\bigr)^{T}\mathbf{R}
\right]
\,\mathrm{d}t,
\\[1mm]
b_{k}^{[j]}
={}&
-\int_{t_{k}}^{t_{k+1}}
x^{T}(t)\mathbf{Q}_{\rho,T}^{[j]}(t)x(t)
\,\mathrm{d}t
\\
&-\left[
x^{T}(t_{k+1})\mathbf{M}x(t_{k+1})
-x^{T}(t_{k})\mathbf{M}x(t_{k})
\right]
\\
&-2\rho\int_{t_{k}}^{t_{k+1}}
x^{T}(t)\mathbf{M}x(t)
\,\mathrm{d}t,
\\[1mm]
r_{k}^{[j]}
={}&
-\left[
x^{T}(t_{k+1})
\Delta_{P}^{[j]}(t_{k+1})
x(t_{k+1})
\right.
\\
&\left.\qquad
-x^{T}(t_{k})
\Delta_{P}^{[j]}(t_{k})
x(t_{k})
\right]
\\
&-2\rho\int_{t_{k}}^{t_{k+1}}
x^{T}(t)\Delta_{P}^{[j]}(t)x(t)
\,\mathrm{d}t
\\
&+2\int_{t_{k}}^{t_{k+1}}
\bigl(\mathbf{z}^{[j]}(t)\bigr)^{T}
\mathbf{R}\Delta_{K}^{[j+1]}(t)x(t)
\,\mathrm{d}t.
\end{align*}
\endgroup
Then \eqref{26} is equivalent to
\begingroup
\setlength{\abovedisplayskip}{6pt}
\setlength{\belowdisplayskip}{6pt}
\setlength{\abovedisplayshortskip}{6pt}
\setlength{\belowdisplayshortskip}{6pt}
\begin{equation}
\boldsymbol{\omega}_{k}^{[j]}\theta^{[j]}
=
b_{k}^{[j]}+r_{k}^{[j]},
\label{209}
\end{equation}
\endgroup
where
\begingroup
\setlength{\abovedisplayskip}{6pt}
\setlength{\belowdisplayskip}{6pt}
\setlength{\abovedisplayshortskip}{6pt}
\setlength{\belowdisplayshortskip}{6pt}
\setlength{\jot}{2pt}
\begin{align*}
\boldsymbol{\omega}_{k}^{[j]}
={}&
\bigl[
\boldsymbol{\omega}_{P,k,1},\ldots,
\boldsymbol{\omega}_{P,k,N_{p}},
\\
&\qquad
\boldsymbol{\omega}_{K,k,1}^{[j]},\ldots,
\boldsymbol{\omega}_{K,k,N_{k}}^{[j]}
\bigr],
\\[1mm]
\theta^{[j]}
={}&
\operatorname{col}\bigl(
\operatorname{vecs}(\mathbf{P}_{1}^{[j]}),\ldots,
\operatorname{vecs}(\mathbf{P}_{N_{p}}^{[j]}),
\\
&\qquad
\operatorname{vec}(\mathbf{K}_{1}^{[j+1]}),\ldots,
\operatorname{vec}(\mathbf{K}_{N_{k}}^{[j+1]})
\bigr).
\end{align*}
\endgroup
Stacking all stored intervals gives
\begingroup
\setlength{\abovedisplayskip}{6pt}
\setlength{\belowdisplayskip}{6pt}
\setlength{\abovedisplayshortskip}{6pt}
\setlength{\belowdisplayshortskip}{6pt}
\begin{equation}
\boldsymbol{\Omega}_{\rho,T}^{[j]}\theta^{[j]}
=
\mathbf{b}_{\rho,T}^{[j]}
+\mathbf{r}_{\rho,T}^{[j]},
\label{29}
\end{equation}
\endgroup
where
\begingroup
\setlength{\abovedisplayskip}{3pt}
\setlength{\belowdisplayskip}{3pt}
\begin{align*}
\boldsymbol{\Omega}_{\rho,T}^{[j]}
&=
\operatorname{col}\left(
\boldsymbol{\omega}_{1}^{[j]},\ldots,
\boldsymbol{\omega}_{N_{d}}^{[j]}
\right),\\
\mathbf{b}_{\rho,T}^{[j]}
&=
\operatorname{col}\left(
b_{1}^{[j]},\ldots,b_{N_{d}}^{[j]}
\right),\\
\mathbf{r}_{\rho,T}^{[j]}
&=
\operatorname{col}\left(
r_{1}^{[j]},\ldots,r_{N_{d}}^{[j]}
\right),
\end{align*}
\endgroup
where
$\boldsymbol{\Omega}_{\rho,T}^{[j]}
\in\mathbb{R}^{N_{d}\times N_{\theta}}$,
$\theta^{[j]}\in\mathbb{R}^{N_{\theta}}$,
and
$\mathbf{b}_{\rho,T}^{[j]}\in\mathbb{R}^{N_{d}}$,
with
$N_{\theta}
=
N_{p}\frac{n(n+1)}{2}+N_{k}mn$. The residual $%
\mathbf{r}_{\rho,T}^{[j]}$ is induced by the approximation errors associated
with the selected basis functions. The following lemma provides an explicit
upper bound on this term.

\begin{lemma}
Fix $T>0$ and $j\geq 0$. Define $\delta _{P}^{[j]}=\sup_{t\in \lbrack
0,T]}\Vert \Delta _{P}^{[j]}(t)\Vert _{F}$, $\delta _{K}^{[j+1]}=\sup_{t\in
\lbrack 0,T]}\Vert \Delta _{K}^{[j+1]}(t)\Vert _{F}$, $X_{T}=\sup_{t\in
\lbrack 0,T]}\Vert x(t)\Vert _{2}$, $Z_{T}^{[j]}=\sup_{t\in \lbrack
0,T]}\Vert \mathbf{z}^{[j]}(t)\Vert _{2}$, and let $h_{\max }=\max_{1\leq
k\leq N_{d}}(t_{k+1}-t_{k})$. Then, for $k=1,\ldots ,N_{d}$, $%
|r_{k}^{[j]}|\leq \eta _{T}^{[j]}$ and
\begingroup
\setlength{\abovedisplayskip}{6pt}
\setlength{\belowdisplayskip}{6pt}
\setlength{\abovedisplayshortskip}{6pt}
\setlength{\belowdisplayshortskip}{6pt}
\begin{equation}
\left\|\mathbf{r}_{\rho,T}^{[j]}\right\|_{2}
\leq
\sqrt{N_{d}}\,\eta_{T}^{[j]},
\label{30}
\end{equation}
\endgroup
where $\eta _{T}^{[j]}=\eta _{P}^{[j]}+\eta _{K}^{[j+1]}$ with
\begingroup
\setlength{\abovedisplayskip}{6pt}
\setlength{\belowdisplayskip}{6pt}
\setlength{\abovedisplayshortskip}{6pt}
\setlength{\belowdisplayshortskip}{6pt}
\setlength{\jot}{2pt}
\begin{align*}
\eta_{P}^{[j]}
&=
2(1+\rho h_{\max})X_{T}^{2}\delta_{P}^{[j]},\\
\eta_{K}^{[j+1]}
&=
2h_{\max}Z_{T}^{[j]}
\left\|\mathbf{R}\right\|_{2}
X_{T}\delta_{K}^{[j+1]}.
\end{align*}
\endgroup
\end{lemma}

\begin{proof}
From the definition of $r_{k}^{[j]}$, the endpoint term satisfies
\begingroup
\setlength{\abovedisplayskip}{6pt}
\setlength{\belowdisplayskip}{6pt}
\setlength{\abovedisplayshortskip}{6pt}
\setlength{\belowdisplayshortskip}{6pt}
\begin{align}
&\left\vert x^{T}(t_{k+1})\Delta
_{P}^{[j]}(t_{k+1})x(t_{k+1})-x^{T}(t_{k})\Delta
_{P}^{[j]}(t_{k})x(t_{k})\right\vert
\notag \\
&\leq \Vert x(t_{k+1})\Vert _{2}^{2}\Vert \Delta _{P}^{[j]}(t_{k+1})\Vert
_{2}+\Vert x(t_{k})\Vert _{2}^{2}\Vert \Delta _{P}^{[j]}(t_{k})\Vert _{2}
\notag \\
&\leq 2X_{T}^{2}\delta _{P}^{[j]}.
\label{31}
\end{align}
\endgroup
Similarly, the integral terms containing $\Delta _{P}^{[j]}(t)$ and $\Delta
_{K}^{[j+1]}(t)$ are bounded, respectively, by $2\rho h_{\max
}X_{T}^{2}\delta _{P}^{[j]}$ and $2h_{\max }Z_{T}^{[j]}\Vert \mathbf{R}\Vert
_{2}X_{T}\delta _{K}^{[j+1]}$. The proof is complete.
\end{proof}

\begin{lemma}
Fix a horizon $T>0$ and an iteration index $j\geq 0$. Let $\mathbf{P}_{\rho
,T}^{[j]}(t)$ be the exact policy-evaluation solution associated with the
current data-driven policy $\widehat{\mathbf{K}}_{\rho ,T}^{[j]}(t)$, and
define $\mathbf{K}_{\rho ,T}^{[j+1]}(t)=\mathbf{R}^{-1}\mathbf{B}^{T}\mathbf{%
P}_{\rho ,T}^{[j]}(t)$. Let $\theta _{0}^{[j]}$ collect the coefficient
matrices $\mathbf{P}_{\ell }^{[j]}$ and $\mathbf{K}_{\ell }^{[j+1]}$ in the
decompositions \eqref{27}-\eqref{28}. Let $\widehat{\theta }^{[j]}=\left(
\boldsymbol{\Omega }_{\rho ,T}^{[j]}\right) ^{\dagger }\mathbf{b}_{\rho
,T}^{[j]}$, and let $\widehat{\mathbf{K}}_{\rho ,T}^{[j+1]}(t)$ be
reconstructed from the gain-coefficient blocks of $\widehat{\theta }^{[j]}$.
If $\func{rank}\left( \boldsymbol{\Omega }_{\rho ,T}^{[j]}\right) =N_{\theta
}$, then
\begingroup
\setlength{\abovedisplayskip}{6pt}
\setlength{\belowdisplayskip}{6pt}
\setlength{\abovedisplayshortskip}{6pt}
\setlength{\belowdisplayshortskip}{6pt}
\begin{equation}
\widehat{\mathbf{K}}_{\rho,T}^{[j+1]}(t)
=
\mathbf{K}_{\rho,T}^{[j+1]}(t)
+\mathcal{E}_{T}^{[j]}(t),
\label{32}
\end{equation}
\endgroup
where
\begingroup
\setlength{\abovedisplayskip}{6pt}
\setlength{\belowdisplayskip}{6pt}
\setlength{\abovedisplayshortskip}{6pt}
\setlength{\belowdisplayshortskip}{6pt}
\begin{equation}
\mathcal{E}_{T}^{[j]}(t)
=
-\sum_{\ell=1}^{N_{k}}
\psi_{\ell}(t)\mathbf{D}_{K,\ell}^{[j]}
-\Delta_{K}^{[j+1]}(t),
\label{33}
\end{equation}
\endgroup
where $\mathbf{D}_{K,\ell }^{[j]}\in \mathbb{R}^{m\times n}$ are the gain
coefficient blocks of $\left( \boldsymbol{\Omega }_{\rho ,T}^{[j]}\right)
^{\dagger }\mathbf{r}_{\rho ,T}^{[j]}.$ Moreover,
\begingroup
\setlength{\abovedisplayskip}{6pt}
\setlength{\belowdisplayskip}{6pt}
\setlength{\abovedisplayshortskip}{6pt}
\setlength{\belowdisplayshortskip}{6pt}
\begin{equation}
\left\|\mathcal{E}_{T}^{[j]}\right\|_{\infty,T}
\leq
\frac{c_{\psi,T}\sqrt{N_{d}}\,\eta_{T}^{[j]}}
{\sigma_{\min}\!\left(\boldsymbol{\Omega}_{\rho,T}^{[j]}\right)}
+\delta_{K}^{[j+1]},
\label{34}
\end{equation}
\endgroup
where $c_{\psi ,T}=\sup_{t\in \lbrack 0,T]}\left( \sum_{\ell =1}^{N_{k}}\psi
_{\ell }^{2}(t)\right) ^{1/2}$.
\end{lemma}

\begin{proof}
By the definition of $\theta _{0}^{[j]}$, the exact coefficients satisfy $%
\boldsymbol{\Omega }_{\rho ,T}^{[j]}\theta _{0}^{[j]}=\mathbf{b}_{\rho
,T}^{[j]}+\mathbf{r}_{\rho ,T}^{[j]}.$ Since $\boldsymbol{\Omega }_{\rho
,T}^{[j]}$ has full column rank, we have $\left( \boldsymbol{\Omega }_{\rho
,T}^{[j]}\right) ^{\dagger }\boldsymbol{\Omega }_{\rho
,T}^{[j]}=I_{N_{\theta }}$. Therefore,
\begingroup
\setlength{\abovedisplayskip}{6pt}
\setlength{\belowdisplayskip}{6pt}
\setlength{\abovedisplayshortskip}{6pt}
\setlength{\belowdisplayshortskip}{6pt}
\begin{align}
\widehat{\theta}^{[j]}-\theta_{0}^{[j]}
&=
\left(\boldsymbol{\Omega}_{\rho,T}^{[j]}\right)^{\dagger}
\left(
\boldsymbol{\Omega}_{\rho,T}^{[j]}\theta_{0}^{[j]}
-\mathbf{r}_{\rho,T}^{[j]}
\right)
-\theta_{0}^{[j]}
\notag\\
&=
-\left(\boldsymbol{\Omega}_{\rho,T}^{[j]}\right)^{\dagger}
\mathbf{r}_{\rho,T}^{[j]}.
\label{35}
\end{align}
\endgroup
Let $\widehat{\mathbf{K}}_{\ell }^{[j+1]}$ and $\mathbf{K}_{\ell }^{[j+1]}$
denote the gain-coefficient blocks of $\widehat{\theta }^{[j]}$ and $\theta
_{0}^{[j]}$, respectively. Then
\begingroup
\setlength{\abovedisplayskip}{6pt}
\setlength{\belowdisplayskip}{6pt}
\setlength{\abovedisplayshortskip}{6pt}
\setlength{\belowdisplayshortskip}{6pt}
\begin{equation}
\widehat{\mathbf{K}}_{\ell}^{[j+1]}
-\mathbf{K}_{\ell}^{[j+1]}
=
-\mathbf{D}_{K,\ell}^{[j]},
\qquad
\ell=1,\ldots,N_{k}.
\label{36}
\end{equation}
\endgroup
Using $\widehat{\mathbf{K}}_{\rho ,T}^{[j+1]}(t)=\sum_{\ell =1}^{N_{k}}\psi
_{\ell }(t)\widehat{\mathbf{K}}_{\ell }^{[j+1]}$ and \eqref{28}, we obtain %
\eqref{32} and \eqref{33}. For every $t\in \lbrack 0,T]$, Cauchy's
inequality gives
\begingroup
\setlength{\abovedisplayskip}{6pt}
\setlength{\belowdisplayskip}{6pt}
\setlength{\abovedisplayshortskip}{6pt}
\setlength{\belowdisplayshortskip}{6pt}
\begin{align}
\left\|
\sum_{\ell=1}^{N_{k}}
\psi_{\ell}(t)\mathbf{D}_{K,\ell}^{[j]}
\right\|_{F}
&\leq
\left(
\sum_{\ell=1}^{N_{k}}\psi_{\ell}^{2}(t)
\right)^{1/2}
\left(
\sum_{\ell=1}^{N_{k}}
\left\|\mathbf{D}_{K,\ell}^{[j]}\right\|_{F}^{2}
\right)^{1/2}
\notag\\
&\leq
c_{\psi,T}
\left\|
\left(
\boldsymbol{\Omega}_{\rho,T}^{[j]}
\right)^{\dagger}
\mathbf{r}_{\rho,T}^{[j]}
\right\|_{2}.
\label{37}
\end{align}
\endgroup
Thus,
\begingroup
\setlength{\abovedisplayskip}{6pt}
\setlength{\belowdisplayskip}{6pt}
\setlength{\abovedisplayshortskip}{6pt}
\setlength{\belowdisplayshortskip}{6pt}
\begin{equation}
\left\|\mathcal{E}_{T}^{[j]}\right\|_{\infty,T}
\leq
c_{\psi,T}
\left\|
\left(
\boldsymbol{\Omega}_{\rho,T}^{[j]}
\right)^{\dagger}
\mathbf{r}_{\rho,T}^{[j]}
\right\|_{2}
+\delta_{K}^{[j+1]}.
\label{38}
\end{equation}
\endgroup
Since $\Vert (\boldsymbol{\Omega }_{\rho ,T}^{[j]})^{\dagger }\Vert
_{2}=1/\sigma _{\min }(\boldsymbol{\Omega }_{\rho ,T}^{[j]})$, \eqref{34}
follows directly from Lemma~3.
\end{proof}

Lemma 4 allows the data-driven update to be viewed as an inexact
finite-horizon PI step. For a fixed horizon $T>0$, define the finite-horizon
policy-improvement mapping as follows. For any bounded matrix-valued
function \(\mathbf K:[0,T]\to\mathbb R^{m\times n}\), let $\mathbf{P}_{\rho ,T}(t;\mathbf{K})$
be the unique solution on \([0,T]\) of
\begingroup
\setlength{\abovedisplayskip}{6pt}
\setlength{\belowdisplayskip}{6pt}
\setlength{\abovedisplayshortskip}{6pt}
\setlength{\belowdisplayshortskip}{6pt}
\begin{align}
-\dot{\mathbf{P}}_{\rho,T}(t;\mathbf{K})
={}&
\bigl(\mathbf{A}_{\rho}-\mathbf{B}\mathbf{K}(t)\bigr)^{T}
\mathbf{P}_{\rho,T}(t;\mathbf{K})
+\mathbf{P}_{\rho,T}(t;\mathbf{K})
\notag\\
&\!\!\!\!\!\!\!\!\!\!\!\!\!\!\times
\bigl(\mathbf{A}_{\rho}-\mathbf{B}\mathbf{K}(t)\bigr)
+\mathbf{Q}
+\mathbf{K}^{T}(t)\mathbf{R}\mathbf{K}(t)
\label{39}
\end{align}
\endgroup
with terminal condition $\mathbf{P}_{\rho ,T}(T;\mathbf{K})=\mathbf{M}$.
Define
\begingroup
\setlength{\abovedisplayskip}{6pt}
\setlength{\belowdisplayskip}{6pt}
\setlength{\abovedisplayshortskip}{6pt}
\setlength{\belowdisplayshortskip}{6pt}
\begin{equation}
\mathcal{F}_{\rho,T}(\mathbf{K})(t)
=
\mathbf{R}^{-1}\mathbf{B}^{T}
\mathbf{P}_{\rho,T}(t;\mathbf{K}),
\qquad
t\in[0,T].
\label{40}
\end{equation}
\endgroup
Then the exact finite-horizon PI recursion can be written as $\mathbf{K}%
_{\rho ,T}^{[j+1]}=\mathcal{F}_{\rho ,T}(\mathbf{K}_{\rho ,T}^{[j]})$.
Moreover, by Lemma~4, the data-driven finite-horizon recursion has the
inexact form
\begingroup
\setlength{\abovedisplayskip}{6pt}
\setlength{\belowdisplayskip}{6pt}
\setlength{\abovedisplayshortskip}{6pt}
\setlength{\belowdisplayshortskip}{6pt}
\begin{equation}
\widehat{\mathbf{K}}_{\rho,T}^{[j+1]}
=
\mathcal{F}_{\rho,T}
\bigl(\widehat{\mathbf{K}}_{\rho,T}^{[j]}\bigr)
+\mathcal{E}_{T}^{[j]}.
\label{41}
\end{equation}
\endgroup

\begin{theorem}
Fix $T>0$, and let $\mathbf{K}_{\rho ,T}^{\ast }(t)$ be defined by \eqref{11}%
. Suppose that the data-driven finite-horizon recursion satisfies \eqref{41}. Then the following statements hold.

\begin{enumerate}
\item[(i)] There exist $a_{T}>0$ and $\bar{\delta}_{T}>0$ such that, if $%
\Vert \mathbf{K}-\mathbf{K}_{\rho ,T}^{\ast }\Vert _{\infty ,T}\leq \bar{%
\delta}_{T}$, then
\begingroup
\setlength{\abovedisplayskip}{6pt}
\setlength{\belowdisplayskip}{6pt}
\setlength{\abovedisplayshortskip}{6pt}
\setlength{\belowdisplayshortskip}{6pt}
\begin{equation}
\left\|
\mathcal{F}_{\rho,T}(\mathbf{K})
-\mathbf{K}_{\rho,T}^{\ast}
\right\|_{\infty,T}
\leq
a_{T}
\left\|
\mathbf{K}
-\mathbf{K}_{\rho,T}^{\ast}
\right\|_{\infty,T}^{2}.
\label{48}
\end{equation}
\endgroup
\item[(ii)] Choose $\delta _{T}\in (0,\bar{\delta}_{T}]$ with $\chi
_{T}=a_{T}\delta _{T}<1$. Assume that, for some $j_{0}\geq 0$, $\left\Vert
\widehat{\mathbf{K}}_{\rho ,T}^{[j_{0}]}-\mathbf{K}_{\rho ,T}^{\ast
}\right\Vert _{\infty ,T}\leq \delta _{T}$ and $\left\Vert \mathcal{E}%
_{T}^{[j]}\right\Vert _{\infty ,T}\leq \bar{\eta}_{T}$, $\bar{\eta}_{T}\leq
(1-\chi _{T})\delta _{T}$, $j\geq j_{0}$. Then, for all $q\geq 0$,
\begingroup
\setlength{\abovedisplayskip}{6pt}
\setlength{\belowdisplayskip}{6pt}
\setlength{\abovedisplayshortskip}{6pt}
\setlength{\belowdisplayshortskip}{6pt}
\setlength{\jot}{2pt}
\begin{align}
&\left\|
\widehat{\mathbf{K}}_{\rho,T}^{[j_{0}+q]}
-\mathbf{K}_{\rho,T}^{\ast}
\right\|_{\infty,T}
\notag\\
&\leq
\chi_{T}^{q}
\left\|
\widehat{\mathbf{K}}_{\rho,T}^{[j_{0}]}
-\mathbf{K}_{\rho,T}^{\ast}
\right\|_{\infty,T}
+
\frac{1-\chi_{T}^{q}}{1-\chi_{T}}
\bar{\eta}_{T}.
\label{49}
\end{align}
\endgroup
Consequently,
\begingroup
\setlength{\abovedisplayskip}{6pt}
\setlength{\belowdisplayskip}{6pt}
\setlength{\abovedisplayshortskip}{6pt}
\setlength{\belowdisplayshortskip}{6pt}
\begin{equation}
\limsup_{j\to\infty}
\left\|
\widehat{\mathbf{K}}_{\rho,T}^{[j]}
-\mathbf{K}_{\rho,T}^{\ast}
\right\|_{\infty,T}
\leq
\frac{\bar{\eta}_{T}}{1-\chi_{T}}.
\label{50}
\end{equation}
\endgroup
\item[(iii)] Under the hypotheses of (ii), if additionally $\Vert \mathcal{E}%
_{T}^{[j]}\Vert _{\infty ,T}\rightarrow 0$, then $\widehat{\mathbf{K}}_{\rho
,T}^{[j]}$ converges uniformly to $\mathbf{K}_{\rho ,T}^{\ast }$ on $[0,T]$.
\end{enumerate}
\end{theorem}

\begin{proof}
Let $\mathbf{K}$ lie in a sufficiently small bounded neighborhood of $%
\mathbf{K}_{\rho ,T}^{\ast }$. Let $\mathbf{S}(t)=\mathbf{P}_{\rho ,T}(t;%
\mathbf{K})-\mathbf{P}_{\rho ,T}^{\ast }(t)$ and $\widetilde{\mathbf{K}}(t)=%
\mathbf{K}(t)-\mathbf{K}_{\rho ,T}^{\ast }(t)$. Since $\mathbf{K}_{\rho
,T}^{\ast }(t)=\mathbf{R}^{-1}\mathbf{B}^{T}\mathbf{P}_{\rho ,T}^{\ast }(t)$%
, the optimal finite-horizon Riccati equation can be rewritten as
\begingroup
\setlength{\abovedisplayskip}{6pt}
\setlength{\belowdisplayskip}{6pt}
\setlength{\abovedisplayshortskip}{6pt}
\setlength{\belowdisplayshortskip}{6pt}
\setlength{\jot}{2pt}
\begin{align}
-\dot{\mathbf{P}}_{\rho,T}^{\ast}(t)
={}&
\bigl(\mathbf{A}_{\rho}-\mathbf{B}\mathbf{K}(t)\bigr)^{T}
\mathbf{P}_{\rho,T}^{\ast}(t)
+\mathbf{P}_{\rho,T}^{\ast}(t)
\notag\\
&\times
\bigl(\mathbf{A}_{\rho}-\mathbf{B}\mathbf{K}(t)\bigr)
+\mathbf{Q}
\notag\\
&+\mathbf{K}^{T}(t)\mathbf{R}\mathbf{K}(t)
-\widetilde{\mathbf{K}}^{T}(t)
\mathbf{R}\widetilde{\mathbf{K}}(t).
\label{51}
\end{align}
\endgroup
Subtracting this equation from the policy-evaluation equation for $\mathbf{P}%
_{\rho ,T}(t;\mathbf{K})$ gives
\begingroup
\setlength{\abovedisplayskip}{6pt}
\setlength{\belowdisplayskip}{6pt}
\setlength{\abovedisplayshortskip}{6pt}
\setlength{\belowdisplayshortskip}{6pt}
\begin{equation}
-\dot{\mathbf{S}}(t)
=
\mathbf{A}_{K}^{T}(t)\mathbf{S}(t)
+\mathbf{S}(t)\mathbf{A}_{K}(t)
+\widetilde{\mathbf{K}}^{T}(t)\mathbf{R}\widetilde{\mathbf{K}}(t),
\label{52}
\end{equation}
\endgroup
with $\mathbf{S}(T)=\mathbf{0}$, where $\mathbf{A}_{K}(t)=\mathbf{A}_{\rho }-%
\mathbf{B}\mathbf{K}(t)$. Let $\boldsymbol{\Phi }_{K}(\tau ,t)$ be the
transition matrix of $\mathbf{A}_{K}(t)$. On the compact interval $[0,T]$
and for $\mathbf{K}$ in a bounded neighborhood of $\mathbf{K}_{\rho
,T}^{\ast }$, there exists $m_{T}>0$ such that
\begingroup
\setlength{\abovedisplayskip}{6pt}
\setlength{\belowdisplayskip}{6pt}
\setlength{\abovedisplayshortskip}{6pt}
\setlength{\belowdisplayshortskip}{6pt}
\begin{equation}
\left\|
\boldsymbol{\Phi}_{K}(\tau,t)
\right\|_{2}
\leq
m_{T},
\qquad
0\leq t\leq \tau\leq T.
\label{53}
\end{equation}
\endgroup
Therefore,
\begingroup
\setlength{\abovedisplayskip}{6pt}
\setlength{\belowdisplayskip}{6pt}
\setlength{\abovedisplayshortskip}{6pt}
\setlength{\belowdisplayshortskip}{6pt}
\setlength{\jot}{2pt}
\begin{align}
\left\|\mathbf{S}(t)\right\|_{F}
&\leq
\int_{t}^{T}
\left\|\boldsymbol{\Phi}_{K}(\tau,t)\right\|_{2}^{2}
\left\|
\widetilde{\mathbf{K}}^{T}(\tau)
\mathbf{R}
\widetilde{\mathbf{K}}(\tau)
\right\|
\,\mathrm{d}\tau
\notag\\
&\leq
T m_{T}^{2}
\left\|\mathbf{R}\right\|_{2}
\left\|
\mathbf{K}-\mathbf{K}_{\rho,T}^{\ast}
\right\|_{\infty,T}^{2}.
\label{54}
\end{align}
\endgroup
It follows that
\begingroup
\setlength{\abovedisplayskip}{6pt}
\setlength{\belowdisplayskip}{6pt}
\setlength{\abovedisplayshortskip}{6pt}
\setlength{\belowdisplayshortskip}{6pt}
\setlength{\jot}{2pt}
\begin{align}
&\left\|
\mathcal{F}_{\rho,T}(\mathbf{K})
-\mathbf{K}_{\rho,T}^{\ast}
\right\|_{\infty,T}
\notag\\
&=
\left\|
\mathbf{R}^{-1}\mathbf{B}^{T}
\left(
\mathbf{P}_{\rho,T}(\cdot;\mathbf{K})
-\mathbf{P}_{\rho,T}^{\ast}(\cdot)
\right)
\right\|_{\infty,T}
\notag\\
&\leq
\left\|
\mathbf{R}^{-1}\mathbf{B}^{T}
\right\|_{2}
T m_{T}^{2}
\left\|\mathbf{R}\right\|_{2}
\left\|
\mathbf{K}-\mathbf{K}_{\rho,T}^{\ast}
\right\|_{\infty,T}^{2}.
\label{55}
\end{align}
\endgroup
Thus \eqref{48} holds with $a_{T}=\left\Vert \mathbf{R}^{-1}\mathbf{B}%
^{T}\right\Vert _{2}Tm_{T}^{2}\Vert \mathbf{R}\Vert _{2}.$ Whenever $\Vert
\widehat{\mathbf{K}}_{\rho ,T}^{[\bar{q}]}-\mathbf{K}_{\rho ,T}^{\ast }\Vert
_{\infty ,T}\leq \delta _{T}$ with $\bar{q}=j_{0}+q$, \eqref{48} implies
\begingroup
\setlength{\abovedisplayskip}{6pt}
\setlength{\belowdisplayskip}{6pt}
\setlength{\abovedisplayshortskip}{6pt}
\setlength{\belowdisplayshortskip}{6pt}
\begin{align}
&\left\|
\widehat{\mathbf{K}}_{\rho,T}^{[\bar{q}+1]}
-\mathbf{K}_{\rho,T}^{\ast}
\right\|_{\infty,T}
\notag\\
&\leq
a_{T}
\left\|
\widehat{\mathbf{K}}_{\rho,T}^{[\bar{q}]}
-\mathbf{K}_{\rho,T}^{\ast}
\right\|_{\infty,T}^{2}
+
\left\|
\mathcal{E}_{T}^{[\bar{q}]}
\right\|_{\infty,T}
\notag\\
&\leq
\chi_{T}
\left\|
\widehat{\mathbf{K}}_{\rho,T}^{[\bar{q}]}
-\mathbf{K}_{\rho,T}^{\ast}
\right\|_{\infty,T}
+\bar{\eta}_{T}.
\label{56}
\end{align}
\endgroup
The inequality above, together with the assumptions
$\Vert \widehat{\mathbf{K}}_{\rho,T}^{[j_0]}
-\mathbf{K}_{\rho,T}^{\ast}\Vert_{\infty,T}
\leq \delta_T$
and
$\bar{\eta}_T\leq(1-\chi_T)\delta_T$,
implies by induction that
$
\left\|
\widehat{\mathbf{K}}_{\rho,T}^{[j_0+q]}
-\mathbf{K}_{\rho,T}^{\ast}
\right\|_{\infty,T}
\leq \delta_T$ for all
$ q\geq 0.
$
Thus, we have
\begingroup
\setlength{\abovedisplayskip}{6pt}
\setlength{\belowdisplayskip}{6pt}
\setlength{\abovedisplayshortskip}{6pt}
\setlength{\belowdisplayshortskip}{6pt}
\setlength{\jot}{2pt}
\begin{align*}
\left\|
\widehat{\mathbf{K}}_{\rho,T}^{[j_0+1]}
-\mathbf{K}_{\rho,T}^{\ast}
\right\|_{\infty,T}
&\leq
\chi_T
\left\|
\widehat{\mathbf{K}}_{\rho,T}^{[j_0]}
-\mathbf{K}_{\rho,T}^{\ast}
\right\|_{\infty,T}
+\bar{\eta}_T,
\\
\left\|
\widehat{\mathbf{K}}_{\rho,T}^{[j_0+2]}
-\mathbf{K}_{\rho,T}^{\ast}
\right\|_{\infty,T}
&\leq
\chi_T^2
\left\|
\widehat{\mathbf{K}}_{\rho,T}^{[j_0]}
-\mathbf{K}_{\rho,T}^{\ast}
\right\|_{\infty,T}
\\
&\quad +(1+\chi_T)\bar{\eta}_T,
\\[-1mm]
&\vdots
\\[-1mm]
\hspace{-1.5em}
\left\|
\widehat{\mathbf{K}}_{\rho,T}^{[\bar q]}
-\mathbf{K}_{\rho,T}^{\ast}
\right\|_{\infty,T}
&\hspace{-1.5em}\leq
\chi_T^q
\left\|
\widehat{\mathbf{K}}_{\rho,T}^{[j_0]}
-\mathbf{K}_{\rho,T}^{\ast}
\right\|_{\infty,T}
\hspace{-0.5em}+\sum_{i=0}^{q-1}\chi_T^i\bar{\eta}_T
\\
&\hspace{-1.5em}=
\chi_T^q
\left\|
\widehat{\mathbf{K}}_{\rho,T}^{[j_0]}
-\mathbf{K}_{\rho,T}^{\ast}
\right\|_{\infty,T}
\hspace{-0.5em}+\frac{1-\chi_T^q}{1-\chi_T}\bar{\eta}_T,
\end{align*}
\endgroup
where $\bar q=j_0+q$. This implies \eqref{49}. Since
$0<\chi_T<1$, taking $q\rightarrow\infty$ yields
\begingroup
\setlength{\abovedisplayskip}{6pt}
\setlength{\belowdisplayskip}{6pt}
\setlength{\abovedisplayshortskip}{6pt}
\setlength{\belowdisplayshortskip}{6pt}
\begin{equation}
\limsup_{q\to\infty}
\left\|
\widehat{\mathbf{K}}_{\rho,T}^{[j_0+q]}
-\mathbf{K}_{\rho,T}^{\ast}
\right\|_{\infty,T}
\leq
\frac{\bar{\eta}_T}{1-\chi_T}.
\label{57}
\end{equation}
\endgroup
If $\Vert \mathcal{E}_{T}^{[j]}\Vert _{\infty ,T}\rightarrow 0$, the same
scalar recursion with $0<\chi _{T}<1$ gives $\Vert \widehat{\mathbf{K}}%
_{\rho ,T}^{[j]}-\mathbf{K}_{\rho ,T}^{\ast }\Vert _{\infty ,T}\rightarrow 0$%
, which is uniform convergence on $[0,T]$.
\end{proof}
\begin{remark}
Equation \eqref{41} identifies the data-driven recursion as a perturbed
finite-horizon policy-improvement iteration. Theorem~5 shows that the local
quadratic behavior of the nominal map $\mathcal F_{\rho,T}$ is retained
modulo the perturbation $\mathcal{E}_{T}^{[j]}$. Thus, the effect of basis truncation and
data conditioning is quantified through $\mathcal{E}_{T}^{[j]}$: uniformly small
perturbations yield practical convergence, whereas vanishing perturbations
recover convergence to $\mathbf{K}_{\rho,T}^{\ast}$.
\end{remark}
\begin{corollary}
Suppose that Assumption~1 holds. Let $\rho >\varepsilon \geq 0$, and choose $%
T\geq T_{\varepsilon }$, where $T_{\varepsilon }$ is given in Theorem~2.
Suppose further that the hypotheses of Theorem~5(ii) hold. Then there exists
$r_{T}>0$ such that
\begingroup
\setlength{\abovedisplayskip}{6pt}
\setlength{\belowdisplayskip}{6pt}
\setlength{\abovedisplayshortskip}{6pt}
\setlength{\belowdisplayshortskip}{6pt}
\begin{equation}
\sup_{\left\|
\mathbf{K}-\mathbf{K}_{\rho,T}^{\ast}(0)
\right\|_{F}\leq r_{T}}
\alpha\!\left(
\mathbf{A}-\mathbf{B}\mathbf{K}
\right)
<-\varepsilon.
\label{58}
\end{equation}
\endgroup
If $\frac{\bar{\eta}_{T}}{1-\chi _{T}}<r_{T}$, then $\widehat{\mathbf{K}}%
_{\rho ,T}^{[j]}(0)$ is $\varepsilon $-admissible for the original system %
\eqref{1} for all sufficiently large $j$.
\end{corollary}

\begin{proof}
Since $T\geq T_{\varepsilon }$, we have $\alpha \!\left( \mathbf{A}-%
\mathbf{B}\mathbf{K}_{\rho ,T}^{\ast }(0)\right) <-\varepsilon .$ By
continuity of the spectral abscissa with respect to the feedback gain, there
exists $r_{T}>0$ such that \eqref{58} holds. On the other hand,
Theorem~5(ii) yields
\begingroup
\setlength{\abovedisplayskip}{6pt}
\setlength{\belowdisplayskip}{6pt}
\setlength{\abovedisplayshortskip}{6pt}
\setlength{\belowdisplayshortskip}{6pt}
\begin{equation}
\limsup_{j\to\infty}
\left\|
\widehat{\mathbf{K}}_{\rho,T}^{[j]}
-\mathbf{K}_{\rho,T}^{\ast}
\right\|_{\infty,T}
\leq
\frac{\bar{\eta}_{T}}{1-\chi_{T}},
\label{59}
\end{equation}
\endgroup
which implies that, for all sufficiently large $j$,
\begingroup
\setlength{\abovedisplayskip}{6pt}
\setlength{\belowdisplayskip}{6pt}
\setlength{\abovedisplayshortskip}{6pt}
\setlength{\belowdisplayshortskip}{6pt}
\begin{equation}
\left\|
\widehat{\mathbf{K}}_{\rho,T}^{[j]}(0)
-\mathbf{K}_{\rho,T}^{\ast}(0)
\right\|_{F}
<
r_{T}.
\label{60}
\end{equation}
\endgroup
Therefore, by \eqref{58}, $\alpha \!\left( \mathbf{A}-\mathbf{B}\widehat{%
\mathbf{K}}_{\rho ,T}^{[j]}(0)\right) <-\varepsilon .$ Hence, $\widehat{%
\mathbf{K}}_{\rho ,T}^{[j]}(0)$ is $\varepsilon $-admissible for all
sufficiently large $j$.
\end{proof}

Corollary 6 gives a theoretical admissibility guarantee, but the
neighborhood radius involved there is model dependent and cannot be checked
directly from data. Therefore, we introduce a data-driven Lyapunov
certificate for the candidate constant gain generated by the finite-horizon
stage. Let $\mathbf{K}_{c}=\widehat{\mathbf{K}}_{\rho ,T}^{[j+1]}(0)$ be a
candidate constant gain generated by the data-driven finite-horizon
recursion. Choose a matrix $\mathbf{W}_{c}=\mathbf{W}_{c}^{T}>\mathbf{0}$.
The purpose of the following test is to determine, using the stored behavior
data only, whether there exists a matrix $\mathbf{P}_{c}=\mathbf{P}_{c}^{T}>%
\mathbf{0}$ such that $(\mathbf{A}-\mathbf{B}\mathbf{K}_{c}+\varepsilon
\mathbf{I})^{T}\mathbf{P}_{c}+\mathbf{P}_{c}(\mathbf{A}-\mathbf{B}\mathbf{K}%
_{c}+\varepsilon \mathbf{I})=-\mathbf{W}_{c}$. If such a matrix $\mathbf{P}%
_{c}$ exists, then $\mathbf{K}_{c}$ is $\varepsilon $-admissible for the
original system.

Define $\mathbf{L}_{c}=\mathbf{B}^{T}\mathbf{P}_{c}\ $and $z_{c}(t)=u_{b}(t)+%
\mathbf{K}_{c}x(t)$. The desired Lyapunov equation is equivalent to
\begingroup
\setlength{\abovedisplayskip}{6pt}
\setlength{\belowdisplayskip}{6pt}
\setlength{\abovedisplayshortskip}{6pt}
\setlength{\belowdisplayshortskip}{6pt}
\begin{equation}
\mathbf{A}^{T}\mathbf{P}_{c}
+\mathbf{P}_{c}\mathbf{A}
=
\mathbf{K}_{c}^{T}\mathbf{L}_{c}
+\mathbf{L}_{c}^{T}\mathbf{K}_{c}
-2\varepsilon\mathbf{P}_{c}
-\mathbf{W}_{c}.
\label{61}
\end{equation}
\endgroup
Along the trajectory $\dot{x}(t)=\mathbf{A}x(t)+\mathbf{B}u_{b}(t),$
one has
\begingroup
\setlength{\abovedisplayskip}{6pt}
\setlength{\belowdisplayskip}{6pt}
\setlength{\abovedisplayshortskip}{6pt}
\setlength{\belowdisplayshortskip}{6pt}
\setlength{\jot}{2pt}
\begin{align}
\frac{\mathrm{d}}{\mathrm{d}t}
\bigl(x^{T}(t)\mathbf{P}_{c}x(t)\bigr)
={}&
-2\varepsilon x^{T}(t)\mathbf{P}_{c}x(t)
\notag\\
&\!\!\!\!\!\!\!\!\!\!\!\!-x^{T}(t)\mathbf{W}_{c}x(t)
+2z_{c}^{T}(t)\mathbf{L}_{c}x(t).
\label{62}
\end{align}
\endgroup
Integrating \eqref{62} over $I_{k}=[t_{k},t_{k+1}]$ gives
\begingroup
\setlength{\abovedisplayskip}{6pt}
\setlength{\belowdisplayskip}{6pt}
\setlength{\abovedisplayshortskip}{6pt}
\setlength{\belowdisplayshortskip}{6pt}
\setlength{\jot}{2pt}
\begin{align}
&x^{T}(t_{k+1})\mathbf{P}_{c}x(t_{k+1})
-x^{T}(t_{k})\mathbf{P}_{c}x(t_{k})
\notag\\
={}&
-2\varepsilon\int_{t_{k}}^{t_{k+1}}
x^{T}(t)\mathbf{P}_{c}x(t)\,\mathrm{d}t
-\int_{t_{k}}^{t_{k+1}}
x^{T}(t)\mathbf{W}_{c}x(t)\,\mathrm{d}t
\notag\\
&+2\int_{t_{k}}^{t_{k+1}}
z_{c}^{T}(t)\mathbf{L}_{c}x(t)\,\mathrm{d}t.
\label{63}
\end{align}
\endgroup
For each interval $I_{k}$, define
\begingroup
\setlength{\abovedisplayskip}{6pt}
\setlength{\belowdisplayskip}{6pt}
\setlength{\abovedisplayshortskip}{6pt}
\setlength{\belowdisplayshortskip}{6pt}
\setlength{\jot}{2pt}
\begin{align*}
\mathbf{d}_{xx,k}
&=
\operatorname{vecv}\bigl(x(t_{k+1})\bigr)
-\operatorname{vecv}\bigl(x(t_{k})\bigr),\\
\mathbf{I}_{xx,k}
&=
\int_{t_{k}}^{t_{k+1}}
\operatorname{vecv}\bigl(x(t)\bigr)\,\mathrm{d}t,\\
\mathbf{I}_{xz,k}(\mathbf{K}_{c})
&=
\int_{t_{k}}^{t_{k+1}}
\bigl(x(t)\otimes z_{c}(t)\bigr)\,\mathrm{d}t.
\end{align*}
\endgroup
Then \eqref{63} can be written as
\begingroup
\setlength{\abovedisplayskip}{6pt}
\setlength{\belowdisplayskip}{6pt}
\setlength{\abovedisplayshortskip}{6pt}
\setlength{\belowdisplayshortskip}{6pt}
\setlength{\jot}{2pt}
\begin{align}
&\bigl(\mathbf{d}_{xx,k}
+2\varepsilon\mathbf{I}_{xx,k}\bigr)^{T}
\operatorname{vecs}(\mathbf{P}_{c})
\notag\\
&\quad
-2\mathbf{I}_{xz,k}^{T}(\mathbf{K}_{c})
\operatorname{vec}(\mathbf{L}_{c})
=
-\mathbf{I}_{xx,k}^{T}
\operatorname{vecs}(\mathbf{W}_{c}).
\label{64}
\end{align}
\endgroup
Stacking all stored intervals yields
\begingroup
\setlength{\abovedisplayskip}{6pt}
\setlength{\belowdisplayskip}{6pt}
\setlength{\abovedisplayshortskip}{6pt}
\setlength{\belowdisplayshortskip}{6pt}
\begin{equation}
\boldsymbol{\Psi}_{c}(\mathbf{K}_{c})\theta_{c}
=
-\mathbf{b}_{c}(\mathbf{W}_{c}),
\label{65}
\end{equation}
\endgroup
where
\begingroup
\setlength{\abovedisplayskip}{6pt}
\setlength{\belowdisplayskip}{6pt}
\setlength{\abovedisplayshortskip}{6pt}
\setlength{\belowdisplayshortskip}{6pt}
\setlength{\jot}{2pt}
\begin{align*}
\theta_{c}
&=
\operatorname{col}\bigl(
\operatorname{vecs}(\mathbf{P}_{c}),
\operatorname{vec}(\mathbf{L}_{c})
\bigr),\\
\boldsymbol{\Psi}_{c}(\mathbf{K}_{c})
&=
\bigl[
\mathbf{D}_{xx}+2\varepsilon\mathbf{I}_{xx},
\;-2\mathbf{I}_{xz}(\mathbf{K}_{c})
\bigr],\\
\mathbf{b}_{c}(\mathbf{W}_{c})
&=
\mathbf{I}_{xx}
\operatorname{vecs}(\mathbf{W}_{c}),
\end{align*}
\endgroup
with \begingroup
\setlength{\abovedisplayskip}{6pt}
\setlength{\belowdisplayskip}{6pt}
\setlength{\abovedisplayshortskip}{6pt}
\setlength{\belowdisplayshortskip}{6pt}
\setlength{\jot}{2pt}
\begin{align*}
\mathbf{D}_{xx}
&=
\func{col}\bigl(
\mathbf{d}_{xx,1}^{T},\ldots,
\mathbf{d}_{xx,N_{d}}^{T}
\bigr),\\
\mathbf{I}_{xx}
&=
\func{col}\bigl(
\mathbf{I}_{xx,1}^{T},\ldots,
\mathbf{I}_{xx,N_{d}}^{T}
\bigr),\\
\mathbf{I}_{xz}(\mathbf{K}_{c})
&=
\func{col}\bigl(
\mathbf{I}_{xz,1}^{T}(\mathbf{K}_{c}),\ldots,
\mathbf{I}_{xz,N_{d}}^{T}(\mathbf{K}_{c})
\bigr).
\end{align*}
\endgroup

The following result gives a data-driven admissibility certificate for the
candidate gain.

\begin{lemma}
Let $\mathbf{K}_{c}$ be a candidate constant gain and let $\mathbf{W}_{c}=%
\mathbf{W}_{c}^{T}>\mathbf{0}$ be given. Suppose that the stored data
satisfy the data-richness condition
\begingroup
\setlength{\abovedisplayskip}{6pt}
\setlength{\belowdisplayskip}{6pt}
\setlength{\abovedisplayshortskip}{6pt}
\setlength{\belowdisplayshortskip}{6pt}
\begin{equation}
\operatorname{rank}
\left[
\mathbf{I}_{xx},\mathbf{I}_{xu}
\right]
=
\frac{n(n+1)}{2}+mn,
\label{66}
\end{equation}
\endgroup
where
\begingroup
\setlength{\abovedisplayskip}{6pt}
\setlength{\belowdisplayskip}{6pt}
\setlength{\abovedisplayshortskip}{6pt}
\setlength{\belowdisplayshortskip}{6pt}
\setlength{\jot}{2pt}
\begin{align*}
\mathbf{I}_{xu}
&=
\func{col}\bigl(
\mathbf{I}_{xu,1}^{T},\ldots,
\mathbf{I}_{xu,N_{d}}^{T}
\bigr),\\
\mathbf{I}_{xu,k}
&=
\int_{t_{k}}^{t_{k+1}}
\bigl(x(t)\otimes u_{b}(t)\bigr)
\,\mathrm{d}t.
\end{align*}
\endgroup If \eqref{65} admits a solution $\theta
_{c} $ whose component satisfies
$\mathbf{P}_{c}
=
\mathbf{P}_{c}^{T}
>
\mathbf{0},$
then $\mathbf{K}_{c}$ is $\varepsilon $-admissible.
\end{lemma}

\begin{proof}
Let $\theta _{c}$ be a solution of \eqref{65}, and let the associated
matrices be $\mathbf{P}_{c}=\mathbf{P}_{c}^{T}$. By the system dynamics, for
any symmetric matrix $\mathbf{P}_{c}$,
\begingroup
\setlength{\abovedisplayskip}{6pt}
\setlength{\belowdisplayskip}{6pt}
\setlength{\abovedisplayshortskip}{6pt}
\setlength{\belowdisplayshortskip}{6pt}
\setlength{\jot}{2pt}
\begin{align}
&x^{T}(t_{k+1})\mathbf{P}_{c}x(t_{k+1})
-x^{T}(t_{k})\mathbf{P}_{c}x(t_{k})
\notag\\
={}&
\int_{t_{k}}^{t_{k+1}}
x^{T}(t)
\bigl(
\mathbf{A}^{T}\mathbf{P}_{c}
+\mathbf{P}_{c}\mathbf{A}
\bigr)
x(t)\,\mathrm{d}t
\notag\\
&+
2\int_{t_{k}}^{t_{k+1}}
u_{b}^{T}(t)\mathbf{B}^{T}\mathbf{P}_{c}x(t)\,\mathrm{d}t.
\label{68}
\end{align}
\endgroup
Combining this identity with \eqref{63} gives, for all stored intervals,
\begingroup
\setlength{\abovedisplayskip}{6pt}
\setlength{\belowdisplayskip}{6pt}
\setlength{\abovedisplayshortskip}{6pt}
\setlength{\belowdisplayshortskip}{6pt}
\setlength{\jot}{2pt}
\begin{align}
&\int_{t_{k}}^{t_{k+1}}
x^{T}(t)\Bigl(
\mathbf{A}^{T}\mathbf{P}_{c}
+\mathbf{P}_{c}\mathbf{A}
+2\varepsilon\mathbf{P}_{c}
+\mathbf{W}_{c}
-\mathbf{K}_{c}^{T}\mathbf{L}_{c}
\notag\\
&\quad
-\mathbf{L}_{c}^{T}\mathbf{K}_{c}
\Bigr)x(t)\,\mathrm{d}t
+2\int_{t_{k}}^{t_{k+1}}
u_{b}^{T}(t)
\bigl(
\mathbf{B}^{T}\mathbf{P}_{c}
-\mathbf{L}_{c}
\bigr)x(t)\,\mathrm{d}t
\notag\\
&=0.
\label{69}
\end{align}
\endgroup
By the rank condition \eqref{66}, it follows that
\begingroup
\setlength{\abovedisplayskip}{6pt}
\setlength{\belowdisplayskip}{6pt}
\setlength{\abovedisplayshortskip}{6pt}
\setlength{\belowdisplayshortskip}{6pt}
\begin{equation}
\bigl(
\mathbf{A}-\mathbf{B}\mathbf{K}_{c}
+\varepsilon\mathbf{I}
\bigr)^{T}\mathbf{P}_{c}
+\mathbf{P}_{c}
\bigl(
\mathbf{A}-\mathbf{B}\mathbf{K}_{c}
+\varepsilon\mathbf{I}
\bigr)
=
-\mathbf{W}_{c}.
\label{70}
\end{equation}
\endgroup
Since $\mathbf{P}_{c}>\mathbf{0}$ and $\mathbf{W}_{c}>\mathbf{0}$, the
Lyapunov theorem implies that $\mathbf{A}-\mathbf{B}\mathbf{K}%
_{c}+\varepsilon \mathbf{I}$ is Hurwitz, which means that $\mathbf{K}_{c}$
is $\varepsilon $-admissible.
\end{proof}

The data-driven bootstrap scheme is summarized in Algorithm 2 and illustrated in Fig. 2. It consists of two nested loops: an inner finite-horizon PI loop that generates a candidate gain at a fixed horizon, and an outer horizon-expansion loop that is activated whenever the candidate fails the data-driven admissibility test. Once the certificate succeeds, the candidate gain is returned as an  initializer for the subsequent infinite-horizon PI.

\begin{remark}
Since the quantities reconstructed from \eqref{27}--\eqref{28} satisfy the
off-policy identity only up to $\mathbf r_{\rho,T}^{[j]}$, a stability test
built on them would inherit basis-truncation errors and would not certify
the constant gain applied to the original system. Therefore, after
$\mathbf K_c=\widehat{\mathbf K}_{\rho,T}^{[j+1]}(0)$ is generated, an
independent Lyapunov variable $\mathbf P_c$ is introduced to verify
admissibility for $\mathbf A-\mathbf B\mathbf K_c+\varepsilon
\mathbf I$.
\end{remark}

\begin{remark}
The proposed data-driven Lyapunov test in Algorithm 2 is both sufficient and necessary in the exact-data setting. This is enabled by introducing an independent Lyapunov variable $\mathbf{P}_{c}$, rather than testing stability using a value matrix inherited from the learning iteration. In contrast, the stopping criteria used in the hybrid and homotopy-based initialization schemes considered in \cite{chen2022homotopic,gao2022resilient} are sufficient stability tests and may therefore certify a gain only after it has already entered the admissible region.
\end{remark}

\begin{algorithm}[t]
\caption{Data-Driven Finite-Horizon Bootstrap Algorithm}
\renewcommand{\algorithmicrequire}{\textbf{Input:}}
\renewcommand{\algorithmicensure}{\textbf{Output:}}

\begin{algorithmic}[1]

\Require
$\varepsilon \geq 0$, $\rho > \varepsilon$,
$\epsilon_{\mathrm{PI}} > 0$,
$T_0 > 0$, $\gamma > 1$,
$\mathbf{M}=\mathbf{M}^{T}\geq\mathbf{0}$,
and $\mathbf{W}_c=\mathbf{W}_c^{T}>\mathbf{0}$.

\Ensure
A certified $\varepsilon$-admissible gain
$\mathbf{K}^{[0]}$.

\State Set $s=0$ and $T_s=T_0$.

\Loop

    \State Collect data on $[0,T_s]$.

    \State Choose a bounded
    $\widehat{\mathbf{K}}_{\rho,T_s}^{[0]}(t)$
    on $[0,T_s]$.

    \Repeat

        \State Construct
        $\boldsymbol{\Omega}_{\rho,T_s}^{[j]}$
        and
        $\mathbf{b}_{\rho,T_s}^{[j]}$
        from \eqref{29}.

        \State Solve
        $\widehat{\boldsymbol{\theta}}^{[j]}
        =
        \left(
        \boldsymbol{\Omega}_{\rho,T_s}^{[j]}
        \right)^{\dagger}
        \mathbf{b}_{\rho,T_s}^{[j]}$.

        \State Set $j \gets j+1$.

    \Until{
        $\left\|
        \widehat{\mathbf{K}}_{\rho,T_s}^{[j]}
        -
        \widehat{\mathbf{K}}_{\rho,T_s}^{[j-1]}
        \right\|_{\infty,T_s}
        \leq \epsilon_{\mathrm{PI}}$
    }

    \State Set $\bar{j}_s=j$ and
    $\mathbf{K}_c
    =
    \widehat{\mathbf{K}}_{\rho,T_s}^{[\bar{j}_s]}(0)$.

    \If{
        \eqref{65} admits a solution with
        $\mathbf{P}_c>\mathbf{0}$
    }

        \State Set
        $\mathbf{K}^{[0]}\gets\mathbf{K}_c$.

        \State \Return $\mathbf{K}^{[0]}$.

    \EndIf

    \State Set
    $T_{s+1}=\gamma T_s$
    and $s\gets s+1$.

\EndLoop

\end{algorithmic}
\end{algorithm}

\begin{figure}[t]
\centering
\includegraphics[width=0.9\linewidth]{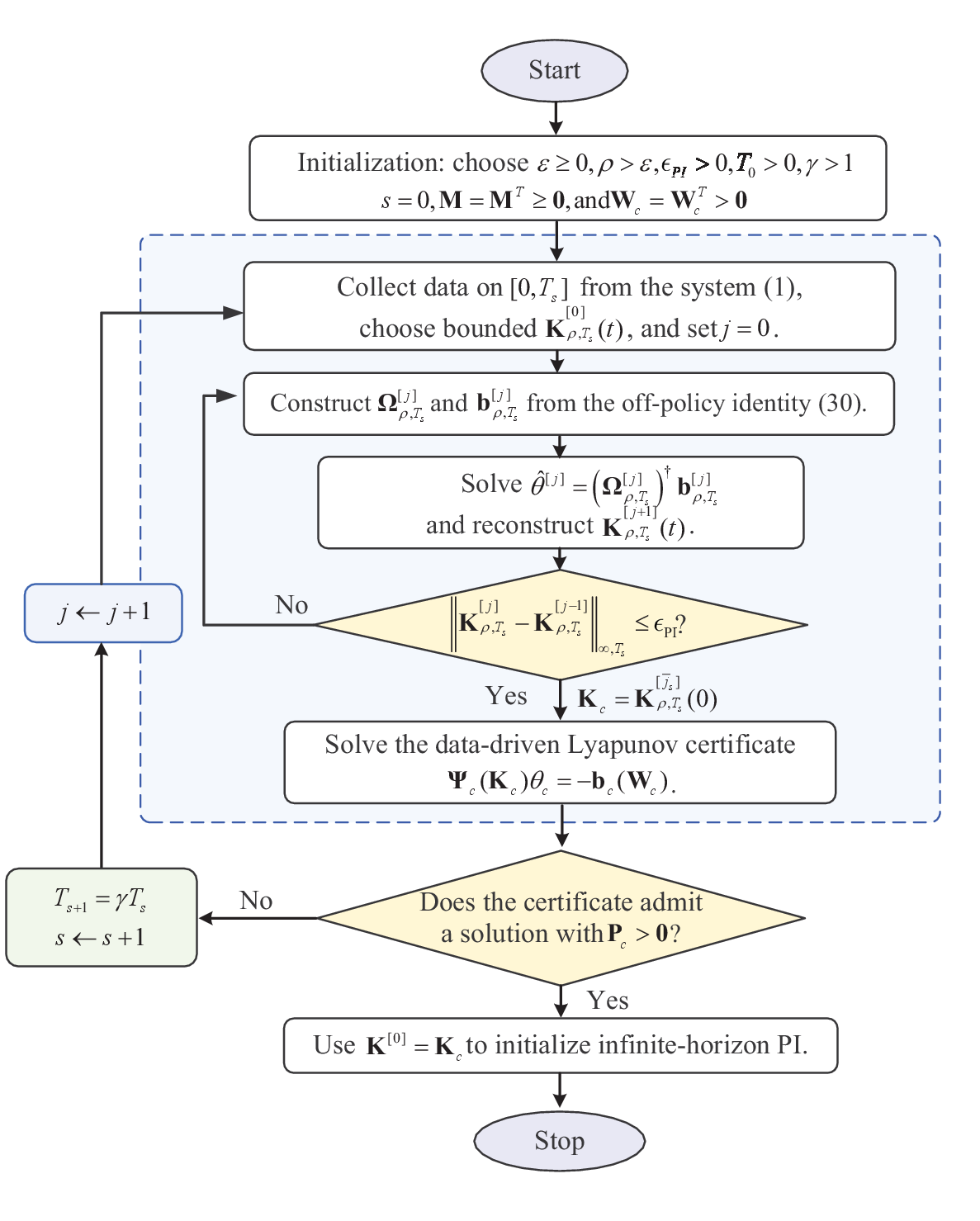}
\caption{Flowchart of the data-driven finite-horizon bootstrap Algorithm 2.}
\end{figure}

\section{Illustrative examples}

\subsection{Case study of the batch reactor system}

In this section, a batch reactor system \cite{walsh2002stability} is used to illustrate the proposed
finite-horizon bootstrap method. The system matrices are given by
\begingroup
\setlength{\abovedisplayskip}{6pt}
\setlength{\belowdisplayskip}{6pt}
\setlength{\abovedisplayshortskip}{6pt}
\setlength{\belowdisplayshortskip}{6pt}

\begin{equation*}
\mathbf{A}
=
\begin{bmatrix}
1.3800  & -0.2077 & 6.7150  & -5.6760 \\
-0.5814 & -4.2900 & 0       & 0.6750 \\
1.0670  & 4.2730  & -6.6540 & 5.8930 \\
0.0480  & 4.2730  & 1.3430  & -2.1040
\end{bmatrix},
\end{equation*}

\begin{equation*}
\mathbf{B}
=
\begin{bmatrix}
0      & 0 \\
5.6790 & 0 \\
1.1360 & -3.1460 \\
1.1360 & 0
\end{bmatrix}.
\end{equation*}

\endgroup
The open-loop poles are $1.9910$, $0.0635$, $-8.6659$, and $-5.0566$, which
shows that the uncontrolled system is unstable. The weighting matrices are
selected as $\mathbf{Q}=\func{diag}(1,1,0.1,0.01)$, $\mathbf{R}=I_{2}$, and
the terminal weight is chosen as $\mathbf{M}=\mathbf{0}$. We first validate
the model-based finite-horizon bootstrap procedure. The design parameters
are set as $\varepsilon =0.2$, $\rho =0.6$, $T_{0}=0.1$, $\epsilon_{\mathrm{PI}}=10^{-7}$ and $\gamma =1.25$.
The finite-horizon policy is initialized by $\mathbf{K}_{\rho
,T_{s}}^{[0]}(t)=\mathbf{0}$ for $t\in \lbrack 0,T_{s}]$, and the horizon is
updated according to $T_{s+1}=\gamma T_{s}$. 

The numerical results are summarized in Table~\ref{tab:batch_bootstrap} and
Fig. 3. It can be seen that, as the finite horizon $T_s$ increases, the
spectral abscissa of the closed-loop matrix $\mathbf{A}-\mathbf{B}\mathbf{K}%
_{c}$ decreases monotonically in this example. When $T_{s}=0.3052$, the
finite-horizon bootstrap produces a gain satisfying $\alpha (\mathbf{A}-%
\mathbf{B}\mathbf{K}_{c})<-0.2. $ Therefore, the algorithm terminates and
the obtained gain is accepted as an $\varepsilon $-admissible initializer.
\begin{table}[t]
\caption{Model-based finite-horizon bootstrap results for the batch reactor
system}
\label{tab:batch_bootstrap}\centering
\begin{tabular}{cccc}
\hline
$s$ & $T_s$ & Iterations & $\alpha(\mathbf{A}-\mathbf{B}\mathbf{K}_c)$ \\
\hline
0 & 0.1000 & 4 & 1.8108 \\
1 & 0.1250 & 4 & 1.6755 \\
2 & 0.1563 & 4 & 1.4501 \\
3 & 0.1953 & 5 & 1.0694 \\
4 & 0.2441 & 5 & 0.3866 \\
5 & 0.3052 & 5 & -0.5822 \\ \hline
\end{tabular}%
\end{table}
The resulting stabilizing initializer is
\begingroup
\setlength{\abovedisplayskip}{6pt}
\setlength{\belowdisplayskip}{6pt}
\setlength{\abovedisplayshortskip}{6pt}
\setlength{\belowdisplayshortskip}{6pt}
\begin{equation*}
\mathbf{K}_{c}
=
\begin{bmatrix}
0.1086  & 0.5935  & 0.1185  & 0.0401 \\
-1.0155 & -0.0248 & -0.6106 & 0.4062
\end{bmatrix}.
\end{equation*}
\endgroup
The corresponding closed-loop poles are $-0.5822\pm 0.5852i$, $-8.3664$, and
$-7.6089$, where $i$ represents the imaginary unit, which confirms that the
proposed finite-horizon bootstrap successfully generates an admissible
initial policy.

\begin{figure}[t]
\centering
\includegraphics[width=1\linewidth]{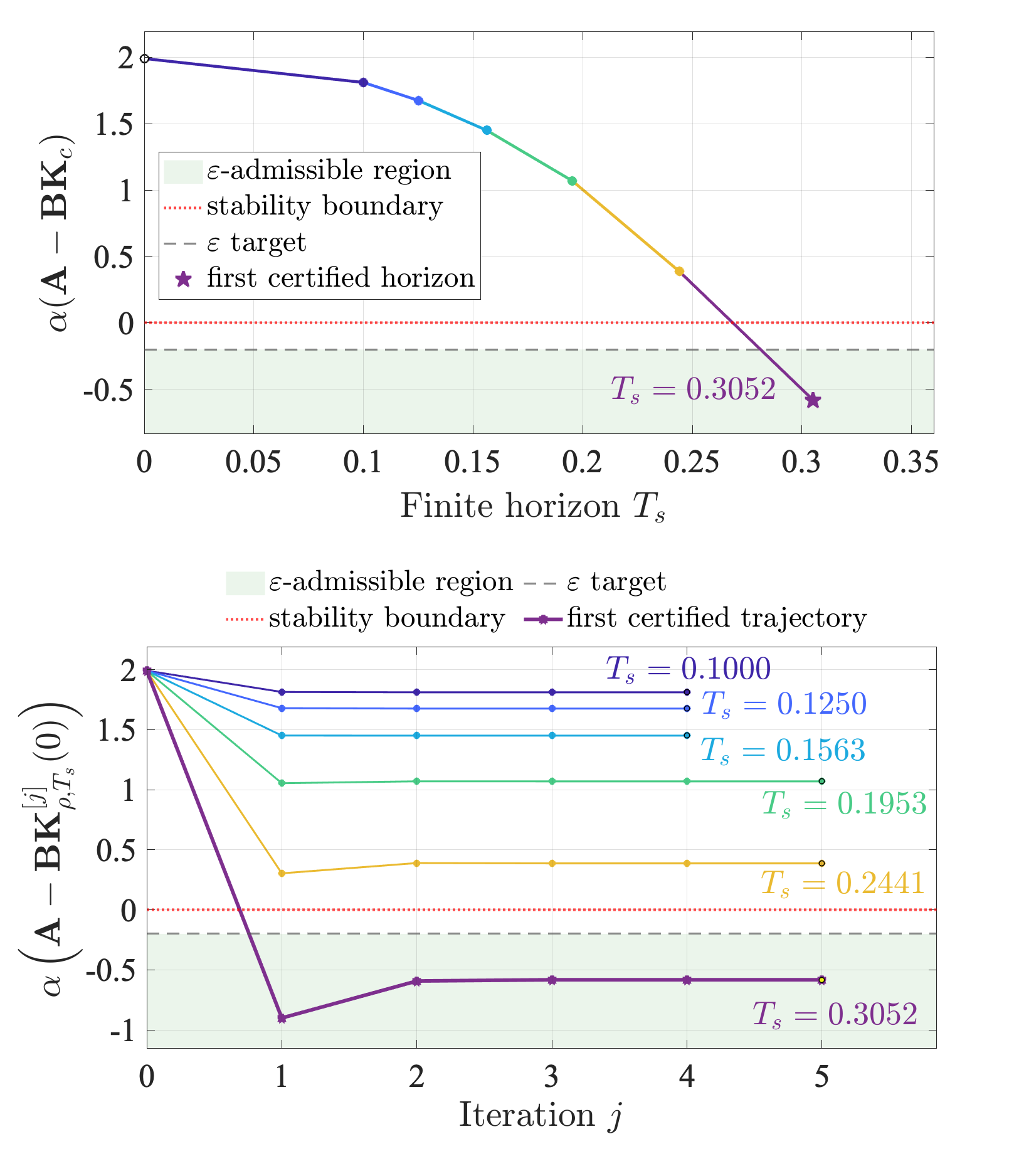}
\caption{ Model-based finite-horizon bootstrap process for the batch reactor
system. The upper panel shows the spectral abscissa $\protect\alpha_s=%
\protect\alpha(\mathbf{A}-\mathbf{B}\mathbf{K}_s)$ of the converged
candidate gain $\mathbf{K}_{\protect\rho,T_s}^{[\bar j_s]}(0)$ as the
horizon is enlarged. The lower panel shows the corresponding inner PI
evolution $\protect\alpha(\mathbf{A}-\mathbf{B}\mathbf{K}_{\protect\rho%
,T_s}^{[j]}(0))$. The first horizon for which the candidate enters the $%
\protect\varepsilon$-admissible region is $T_s=0.3052$. }
\end{figure}

We next evaluate the data-driven bootstrap procedure in Algorithm~2. The
certificate matrix is chosen as $\mathbf{W}_{c}=I_{4}$. The initial states
are randomly generated with norms in $[0.3,0.7]$. The basis functions
\eqref{27}--\eqref{28} are selected as
\begingroup
\setlength{\abovedisplayskip}{6pt}
\setlength{\belowdisplayskip}{6pt}
\setlength{\abovedisplayshortskip}{6pt}
\setlength{\belowdisplayshortskip}{6pt}
\setlength{\jot}{2pt}
\begin{align*}
\phi_{\ell}(t)
&=
\left(1-\frac{t}{T_{s}}\right)
\left(\frac{t}{T_{s}}\right)^{\ell-1},\\
\psi_{\ell}(t)
&=
\left(\frac{t}{T_{s}}\right)^{\ell-1},
\qquad
\ell=1,\ldots,8.
\end{align*}
\endgroup
Thus $N_{p}=N_{k}=8$. The data-driven PI loop is initialized by $\widehat{%
\mathbf{K}}_{\rho ,T_{s}}^{[0]}(t)=\mathbf{0}$.
The data-driven results are shown in Fig.~4. For all tested horizons, the
data-richness condition for the Lyapunov certificate is satisfied with $%
\func{rank}[\mathbf{I}_{xx},\mathbf{I}_{xu}]=18$, and the finite-horizon
regression matrix has full column rank, i.e., $\func{rank}(\boldsymbol{\Omega%
}_{\rho,T_s}^{[j]})=144$. For numerical implementation, the linear equation %
\eqref{65} is solved in the least-squares sense. The relative residual
\begingroup
\setlength{\abovedisplayskip}{6pt}
\setlength{\belowdisplayskip}{6pt}
\setlength{\abovedisplayshortskip}{6pt}
\setlength{\belowdisplayshortskip}{6pt}
\begin{equation}
\delta_{c}
=
\frac{
\left\|
\widetilde{\boldsymbol{\Psi}}_{c}\widehat{\theta}_{c}
+\widetilde{\mathbf{b}}_{c}
\right\|_{2}
}{
\max\left\{
1,
\left\|\widetilde{\mathbf{b}}_{c}\right\|_{2}
\right\}
}
\end{equation}
\endgroup
is reported only as a numerical accuracy indicator, where $\widetilde{%
\boldsymbol{\Psi}}_c$ and $\widetilde{\mathbf{b}}_c$ denote the row-scaled
data matrix and right-hand side. As $T_s$ increases, the spectral abscissa
of the candidate closed-loop matrix moves toward the $\varepsilon$%
-admissible region. The candidates generated for the first five horizons are
rejected by the data-driven Lyapunov certificate. At $T_s=0.3052$, the
data-driven finite-horizon PI converges in five iterations and produces
\begin{equation*}
\mathbf{K}_{c}=
\begin{bmatrix}
0.1085 & 0.5936 & 0.1184 & 0.0402 \\
-1.0155 & -0.0247 & -0.6106 & 0.4062%
\end{bmatrix}%
.
\end{equation*}
For this gain, the Lyapunov certificate \eqref{65} has relative residual $%
7.432\times 10^{-4}$ and $\lambda_{\min}(\mathbf{P}_c)=5.938\times 10^{-2}>0$%
. Therefore, the certificate succeeds and Algorithm~2 terminates. Thus the
gain certified from data is an $\varepsilon$-admissible initializer for the
subsequent infinite-horizon PI. Using $\mathbf{K}_c$ as the initial policy,
the subsequent off-policy RL converges to the optimal solution $\mathbf{K}%
^\ast$ within a few iterations. As shown in Fig. 5, both $\|\mathbf{P}^{[j]}-%
\mathbf{P}^\ast\|$ and $\|\mathbf{K}^{[j]}-\mathbf{K}^\ast\|$ decrease
rapidly, confirming that the bootstrap gain provides a valid initializer for
the infinite-horizon data-driven PI stage.

\begin{figure}[t]
\centering
\includegraphics[width=1\linewidth]{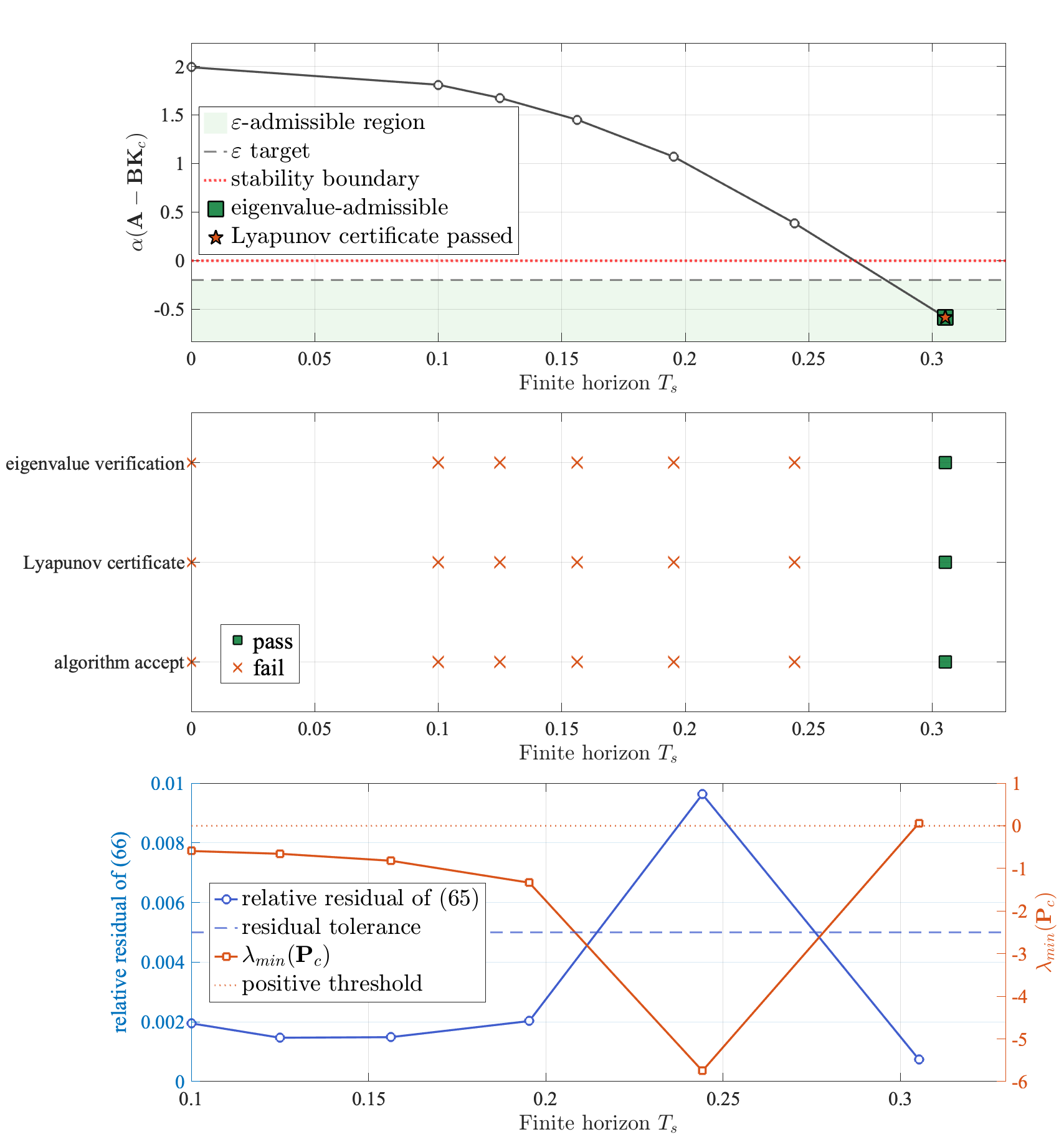}
\caption{Data-driven finite-horizon bootstrap process for the batch reactor
system. The upper panel shows the spectral abscissa of the candidate
closed-loop matrix as the horizon is enlarged. The middle panel compares the
eigenvalue-based verification, the data-driven Lyapunov certificate, and the
final algorithmic decision. The lower panel reports the residual and
positivity conditions of the Lyapunov certificate. }
\end{figure}
\begin{figure}[t]
\centering
\includegraphics[width=1\linewidth]{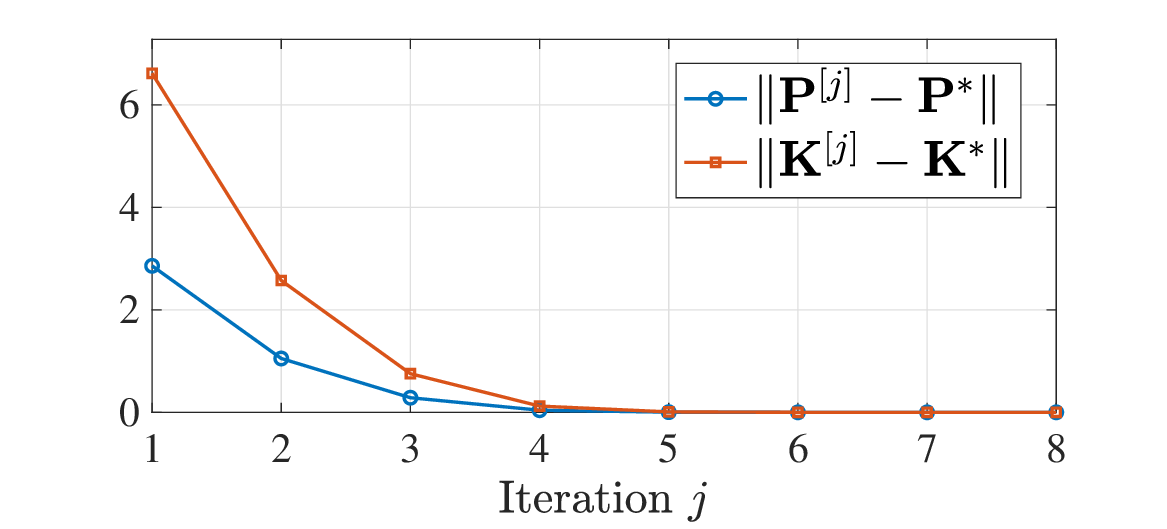}
\caption{Convergence of the off-policy RL initialized by the bootstrap gain.}
\end{figure}

\subsection{Comparative study of the two-mass-spring system}

We next consider the two-mass-spring system used in the homotopic PI
study \cite{chen2022homotopic} to compare Algorithm~2 with existing data-driven RL
methods under the same plant and cost.  The system matrices are
\begin{equation*}
\mathbf{A}=%
\begin{bmatrix}
0 & 1 & 0 & 0 \\
-\frac{k_{1}+k_{2}}{m_{1}} & 0 & \frac{k_{2}}{m_{1}} & 0 \\
0 & 0 & 0 & 1 \\
\frac{k_{2}}{m_{2}} & 0 & -\frac{k_{2}}{m_{2}} & 0%
\end{bmatrix}%
,\quad \mathbf{B}=%
\begin{bmatrix}
0 \\
\frac{1}{m_{1}} \\
0 \\
0%
\end{bmatrix}%
.
\end{equation*}%
The physical parameters are $k_{1}=1.5~\mathrm{N/m}$, $k_{2}=1~\mathrm{N/m}$%
, $m_{1}=1.1~\mathrm{kg}$, and $m_{2}=0.9~\mathrm{kg}$. The weights are
selected as $\mathbf{Q}=\func{diag}(1,1,2,1)$ and $\mathbf{R}=I_{1}$.
\begin{figure}[t]
\centering
\includegraphics[width=1\linewidth]{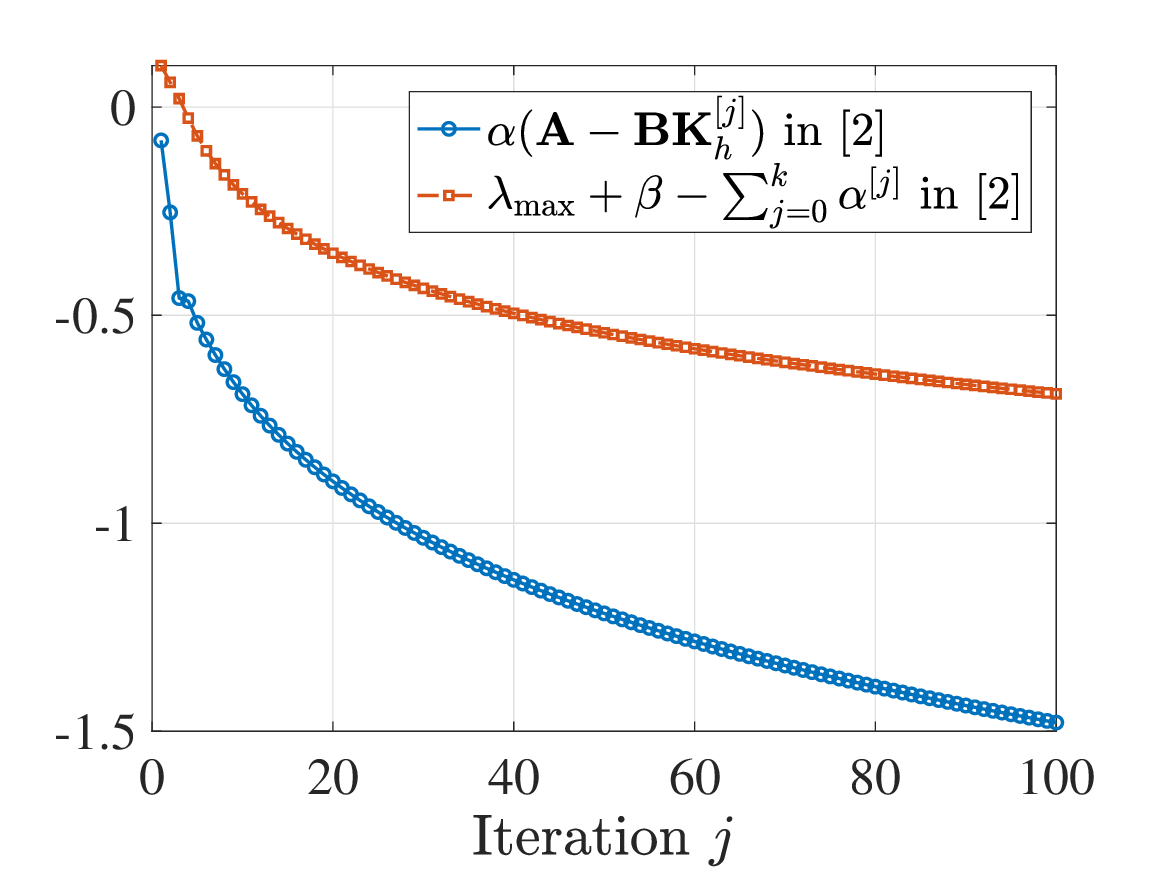}
\caption{Convergence of the stabilizing stage in the data-driven homotopic PI \cite{chen2022homotopic}. The curves show the closed-loop spectral abscissa and the auxiliary stability index used in the iteration.}
\end{figure}
\begin{figure}[t]
\centering
\includegraphics[width=0.9\linewidth]{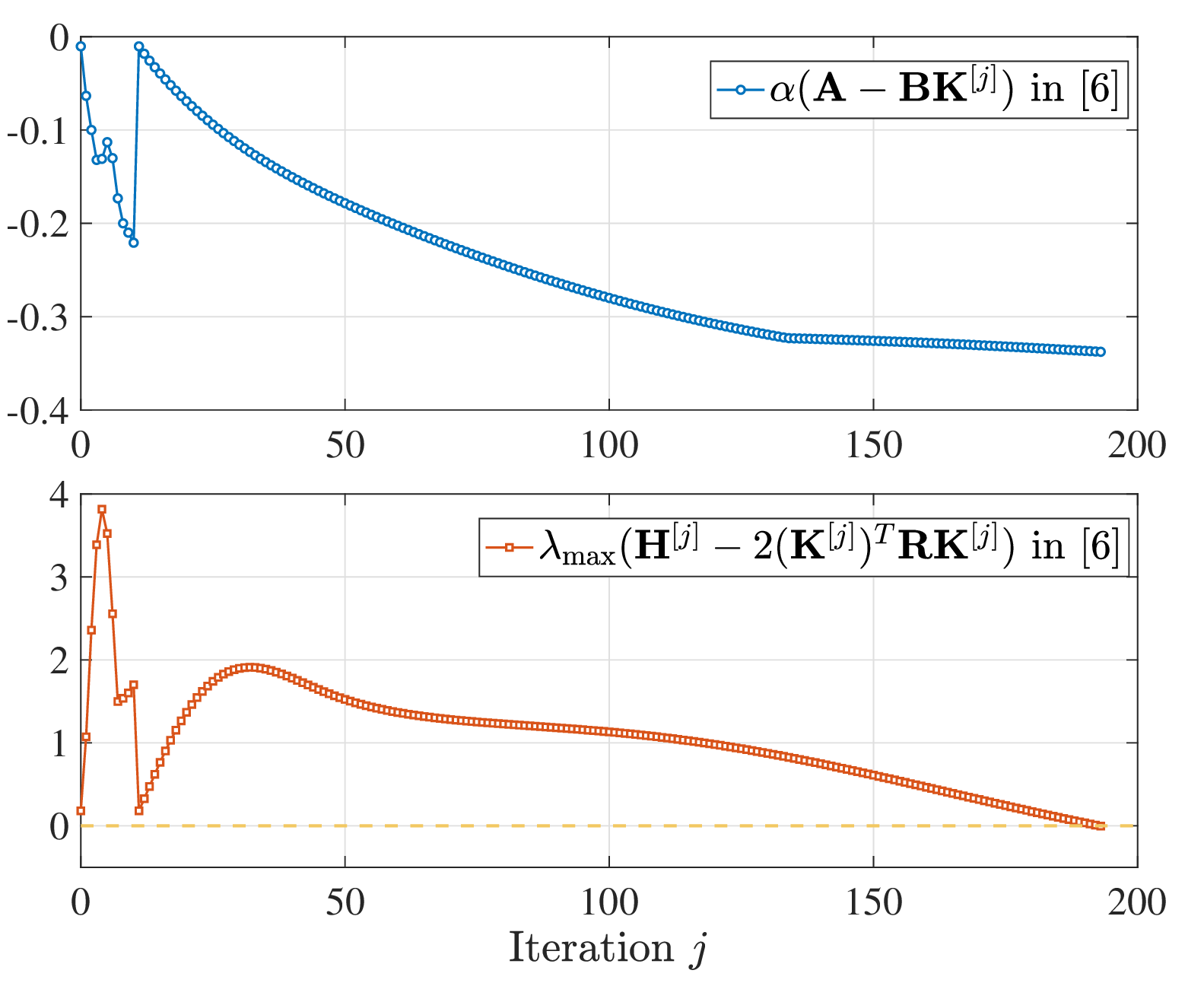}
\caption{Convergence of the VI stage in the data-driven hybrid PI \cite{gao2022resilient}. The upper panel shows the closed-loop spectral abscissa, and the lower panel shows the data-driven Lyapunov certificate quantity.}
\end{figure}
\begin{figure}[t]
\centering
\includegraphics[width=1\linewidth]{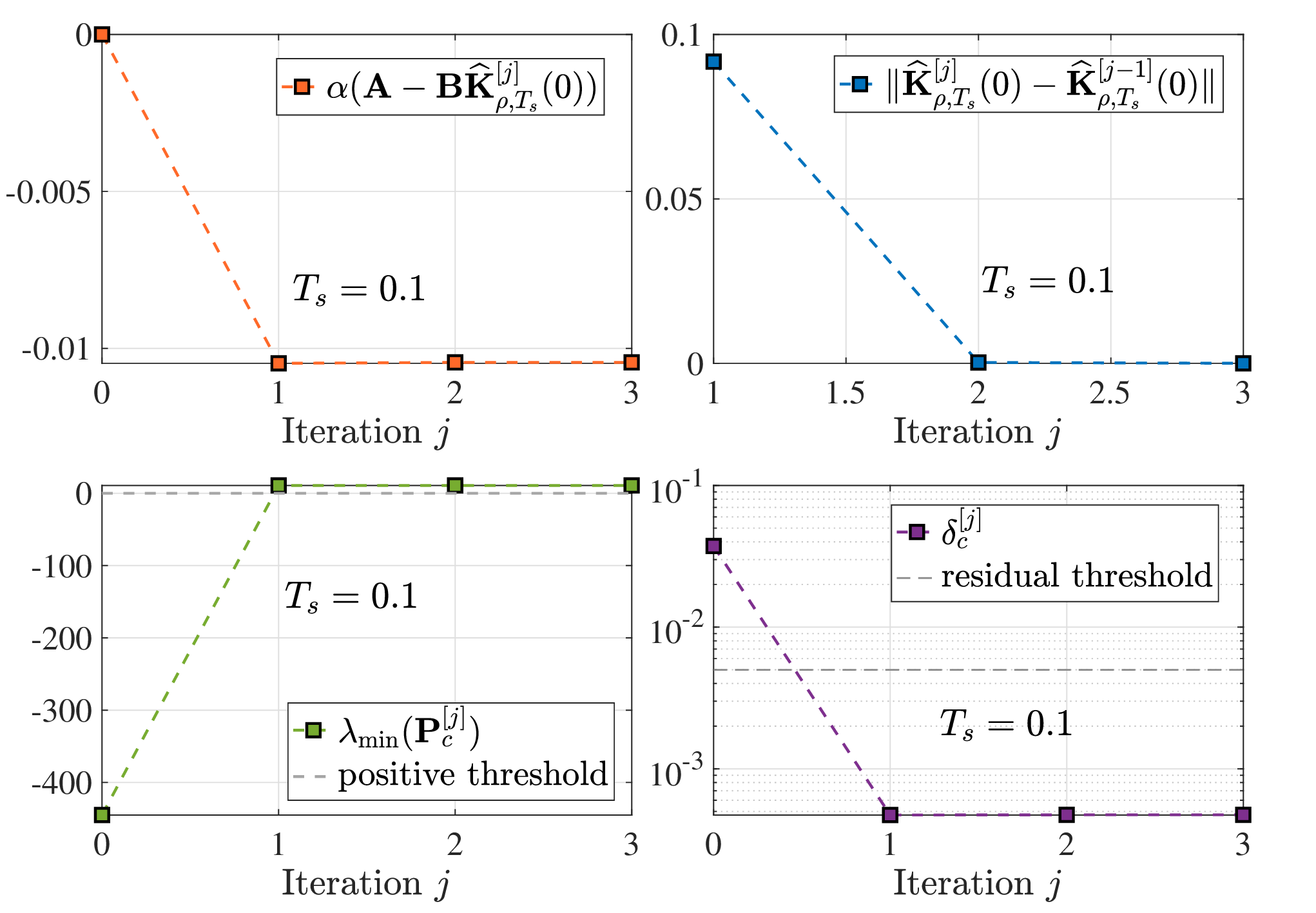}
\caption{Data-driven finite-horizon bootstrap process of Algorithm 2. }
\end{figure}
For the existing data-driven homotopic PI method \cite{chen2022homotopic}, the reported simulation
uses $\mathbf{K}^{[0]}=\mathbf{0}$ and $\beta -\alpha ^{[0]}=0.1$. The data
are collected over $[0,5]$ from the initial condition $%
x(0)=[0.2,0.2,0.2,0.2]^{T} $, with a probing input formed by sine and cosine
signals of different frequencies. As shown in Fig. 6, the learned
stabilizing gain is
\begingroup
\setlength{\abovedisplayskip}{6pt}
\setlength{\belowdisplayskip}{6pt}
\setlength{\abovedisplayshortskip}{6pt}
\setlength{\belowdisplayshortskip}{6pt}
\begin{equation*}
\mathbf{K}_{h}^{[101]}
=
\begin{bmatrix}
14.7218 & 6.5855 & -2.4442 & 16.4721
\end{bmatrix},
\end{equation*}
\endgroup
which gives the closed-loop poles $-1.51455\pm 1.66164i$ and $-1.47887\pm
0.752358i$.

For comparison, the hybrid PI method \cite{gao2022resilient} uses $\mathbf{P}^{[0]}=0.1I_{4}$, $\epsilon _{j}=1/(j+2)$, and $\mathcal{B}%
_{q}=\{\mathbf{P}>\mathbf{0}:\Vert \mathbf{P}\Vert \leq 5(q+1)\}$, $%
q=0,1,\ldots $. At each VI iteration, $\mathbf{H}^{\left[ j\right] }$ and $%
\mathbf{K}^{\left[ j\right] }_{v}$ are recovered from the data equation, and $%
\widetilde{\mathbf{P}}^{\left[ j+1\right] }$ is generated by the VI update.
If $\widetilde{\mathbf{P}}^{\left[ j+1\right] }\notin \mathcal{B}_{q}$, the
iteration is reset to $\mathbf{P}^{[0]}$ and $q$ is increased by one.
Otherwise, the VI stage is terminated once the data-driven Lyapunov
certificate $\mathbf{P}^{\left[ j\right] }>\mathbf{0}$ and $\mathbf{H}^{%
\left[ j\right] }-2(\mathbf{K}^{[j]}_{v})^{T}\mathbf{R}\mathbf{K}^{\left[ j%
\right] }_{v}<\mathbf{0}$ is satisfied, where $\mathbf{H}^{\left[ j\right] }=%
\mathbf{A}^{T}\mathbf{P}^{\left[ j\right] }+\mathbf{P}^{\left[ j\right] }%
\mathbf{A}$. Fig. 7 shows the VI-stage learning process. After one
reset caused by the bounded-set test, the spectral abscissa
$\alpha(\mathbf{A}-\mathbf{B}\mathbf{K}^{[j]}_{v})$ remains negative and decreases
gradually, as shown in Fig. 7. Meanwhile,
$\lambda_{\max}(\mathbf{H}^{[j]}-2(\mathbf{K}^{[j]}_{v})^{T}\mathbf{R}\mathbf{K}^{\left[ j%
\right] }_{v})$
crosses zero at the $193$rd VI iteration, which activates the data-driven
termination certificate. The certified
initializer is
\[
\mathbf{K}_{v}^{[193]}
=\begin{bmatrix}
0.7714 & 1.5364 & -0.1604 & 0.9992
\end{bmatrix},
\]
and the corresponding closed-loop poles are
$-0.3376\pm1.6540i$ and $-0.3607\pm0.7862i$. 

For Algorithm~2, we use $\varepsilon =0$, $\rho =0.1$, $T_{0}=0.1$, $\gamma
=1.25$, $\mathbf{M}=\mathbf{0}$, and $\mathbf{W}_{c}=I_{4}$. The
finite-horizon policy is initialized by $\widehat{\mathbf{K}}_{\rho
,T_{s}}^{[0]}(t)=\mathbf{0}$, and the same basis functions and stopping
tolerance $10^{-7}$ as in the preceding Section 5.1 are used. As shown in Fig. 8, the
data-driven bootstrap procedure succeeds at the first tested horizon $T_s=0.1
$. The resulting candidate gain is
\begingroup
\setlength{\abovedisplayskip}{6pt}
\setlength{\belowdisplayskip}{6pt}
\setlength{\abovedisplayshortskip}{6pt}
\setlength{\belowdisplayshortskip}{6pt}
\begin{equation*}
\mathbf{K}_{c}
=
\begin{bmatrix}
-0.0057 & 0.0911 & 0.0041 & 0.0003
\end{bmatrix}.
\end{equation*}
\endgroup
For this gain, the data-driven Lyapunov certificate gives $\lambda_{\min}(%
\mathbf{P}_c)=10.84354>0$ and a relative residual $4.732\times10^{-4}$,
which is below the prescribed tolerance. The closed-loop spectral abscissa
is $\alpha(\mathbf{A}-\mathbf{B}\mathbf{K}_{\mathrm{c}})
=-1.0445\times10^{-2}<0$, confirming that the certified initializer is
admissible. The results in Figs. 6-8 illustrate the distinct advantage of the proposed bootstrap mechanism. 
The homotopic PI \cite{chen2022homotopic} and hybrid PI \cite{gao2022resilient} both succeed in generating admissible initializers, but their certification requires $101$ stabilizing iterations and $193$ iterations, respectively. 
By contrast, Algorithm 2 succeeds at the first tested horizon $T_s=0.1$ and the finite-horizon PI converges in only three iterations.

\section{Conclusion}

This article develops a finite-horizon bootstrap method for
data-driven infinite-horizon PI without requiring an
initial stabilizing policy. By exploiting the well-posedness of
finite-horizon policy evaluation and enlarging the horizon for a
shifted system, the method generates an admissible candidate gain.
An off-policy implementation reconstructs the time-varying value
matrix and policy from measured trajectories, while an independent
data-driven Lyapunov certificate verifies the candidate gain before
it initializes infinite-horizon PI. Numerical studies
demonstrate the effectiveness of the proposed method.


\vspace{1.5ex}
\bibliographystyle{abbrv}
\bibliography{shwj-workwjc}

@string{ auto = {Automatica} }

@string{ ieeeac = {IEEE Trans. Autom. Control} }

@string{ ieeease = {IEEE Trans. Automation Sci. Eng.} }

@string{ ieeec = {IEEE Trans. Cybern.} }

@string{ ieeecaa = {IEEE/CAA  J. Automatica Sinica} }

@string{ ieeecst = {IEEE Trans. Control Syst. Technol.} }

@string{ ieeeie = {IEEE Trans. Ind. Electron.} }

@string{ ieeennls = {IEEE Trans. Neural Netw. Learn. Syst.} }

@string{ ieeesmcs = {IEEE Trans. Syst. Man, Cybern. Syst.} }

@string{ m = {Mechatronics} }

@string{ ma= {Math. Ann.} }

@string{ ieeevt = {IEEE Trans. Veh. Technol.} }

@article{kleinman1968iterative,
  title={On an iterative technique for {R}iccati equation computations},
  author={Kleinman, David},
  journal=ieeeac,
  volume={13},
  number={1},
  pages={114--115},
  year={1968},
  month={Feb.},
  publisher={IEEE}
}

@article{jiang2012computational,
  title={Computational adaptive optimal control for continuous-time linear systems with completely unknown dynamics},
  author={Jiang, Yu and Jiang, Zhong-Ping},
  journal=auto,
  volume={48},
  number={10},
  pages={2699--2704},
  year={2012},
  publisher={Elsevier}
}

@article{vrabie2009adaptive,
  title={Adaptive optimal control for continuous-time linear systems based on policy iteration},
  author={Vrabie, Draguna and Pastravanu, Octavian and Abu-Khalaf, Murad and Lewis, Frank L},
  journal=auto,
  volume={45},
  number={2},
  pages={477--484},
  year={2009},
  publisher={Elsevier}
}

@article{liu2025adaptive,
  title={Adaptive dynamic programming-regulated extremum seeking for distributed feedback optimization},
  author={Liu, Tong and Krsti{\'c}, Miroslav and Jiang, Zhong-Ping},
  journal=ieeeac,
  volume={70},
  number={11},
  pages={7675--7682},
  month={Nov.},
  year={2025},
  publisher={IEEE}
}

@article{cui2025learning,
  title={Learning-based adaptive optimal control of linear time-delay systems: A value iteration approach},
  author={Cui, Leilei and Pang, Bo and Krsti{\'c}, Miroslav and Jiang, Zhong-Ping},
  journal=auto,
  volume={171},
  pages={111944},
  year={2025},
  publisher={Elsevier}
}

@book{lewis2012optimal,
  title={Optimal {C}ontrol},
  author={Lewis, Frank L and Vrabie, Draguna and Syrmos, Vassilis L},
  year={2012},
  publisher={John Wiley \& Sons}
}

@article{shen2024data,
  title={Data-driven near optimization for fast sampling singularly perturbed systems},
  author={Shen, Hao and Peng, Chuanjun and Yan, Huaicheng and Xu, Shengyuan},
  journal=ieeeac,
  volume={69},
  number={7},
  pages={4689--4694},
  month={Jul.},
  year={2024},
  publisher={IEEE}
}

@article{gao2016adaptive,
  title={Adaptive dynamic programming and adaptive optimal output regulation of linear systems},
  author={Gao, Weinan and Jiang, Zhong-Ping},
  journal=ieeeac,
  volume={61},
  number={12},
  pages={4164--4169},
  month={Dec.},
  year={2016},
  publisher={IEEE}
}

@article{possieri2022value,
  title={Value iteration for continuous-time linear time-invariant systems},
  author={Possieri, Corrado and Sassano, Mario},
  journal=ieeeac,
  volume={68},
  number={5},
  pages={3070--3077},
  month={May},
  year={2023},
  publisher={IEEE}
}

@article{wei2016value,
  title={Value iteration adaptive dynamic programming for optimal control of discrete-time nonlinear systems},
  author={Wei, Qinglai and Liu, Derong and Lin, Hanquan},
  journal=ieeec,
  volume={46},
  number={3},
  pages={840--853},
  month={Mar.},
  year={2016},
  publisher={IEEE}
}

@article{luo2019balancing,
  title={Balancing value iteration and policy iteration for discrete-time control},
  author={Luo, Biao and Yang, Yin and Wu, Huai-Ning and Huang, Tingwen},
  journal=ieeesmcs,
  volume={50},
  number={11},
  pages={3948--3958},
  month={Nov.},
  year={2020},
  publisher={IEEE}
}

@article{pang2020adaptive,
  title={Adaptive optimal control of linear periodic systems: An off-policy value iteration approach},
  author={Pang, Bo and Jiang, Zhong-Ping},
  journal=ieeeac,
  volume={66},
  number={2},
  pages={888--894},
  month={Feb.},
  year={2021},
  publisher={IEEE}
}

@article{bian2016value,
  title={Value iteration and adaptive dynamic programming for data-driven adaptive optimal control design},
  author={Bian, Tao and Jiang, Zhong-Ping},
  journal=auto,
  volume={71},
  pages={348--360},
  year={2016},
  publisher={Elsevier}
}

@article{gao2022resilient,
  title={Resilient reinforcement learning and robust output regulation under denial-of-service attacks},
  author={Gao, Weinan and Deng, Chao and Jiang, Yi and Jiang, Zhong-Ping},
  journal=auto,
  volume={142},
  pages={110366},
  year={2022},
  publisher={Elsevier}
}

@article{shen2024secure,
  title={Secure control for {M}arkov jump cyber-physical systems subject to malicious attacks: A resilient hybrid learning scheme},
  author={Shen, Hao and Wang, Yun and Wu, Jiacheng and Park, Ju H and Wang, Jing},
  journal=ieeec,
  volume={54},
  number={11},
  pages={7068--7079},
   month={Nov.},
  year={2024},
  publisher={IEEE}
}

@article{liang2025cooperative,
  title={Cooperative adaptive cruise control of connected and autonomous vehicles via hybrid iteration},
  author={Liang, Xue and Gao, Weinan and Hu, Chuan and Chai, Tianyou},
  journal=ieeevt,
  volume={75},
  number={6},
  pages={8793--8804},
   month={Jun.},
  year={2026},
  publisher={IEEE}
}

@article{jiang2022bias,
  title={Bias-policy iteration based adaptive dynamic programming for unknown continuous-time linear systems},
  author={Jiang, Huaiyuan and Zhou, Bin},
  journal=auto,
  volume={136},
  pages={110058},
  year={2022},
  publisher={Elsevier}
}

@article{yang2021model,
  title={Model-free $\lambda$-policy iteration for discrete-time linear quadratic regulation},
  author={Yang, Yongliang and Kiumarsi, Bahare and Modares, Hamidreza and Xu, Chengzhong},
  journal=ieeennls,
  volume={34},
  number={2},
  pages={635--649},
  month={Feb.},
  year={2023},
  publisher={IEEE}
}

@article{shen2025data,
  title={Data-driven single-loop policy iteration control of uncertain singularly perturbed systems},
  author={Shen, Hao and Wang, Yun and Yan, Huaicheng and Xu, Shengyuan},
  journal=ieeeac,
  volume={70},
  number={12},
  pages={8314--8320},
  month={Dec.},
  year={2025},
  publisher={IEEE}
}

@article{zhao2024novel,
  title={Novel single-loop policy iteration for linear zero-sum games},
  author={Zhao, Jianguo and Yang, Chunyu and Gao, Weinan and Park, Ju H},
  journal=auto,
  volume={163},
  pages={111551},
  year={2024},
  publisher={Elsevier}
}

@article{chen2022homotopic,
  title={Homotopic policy iteration-based learning design for unknown linear continuous-time systems},
  author={Chen, Ci and Lewis, Frank L and Li, Bo},
  journal=auto,
  volume={138},
  pages={110153},
  year={2022},
  publisher={Elsevier}
}

@article{fan2025homotopy,
  title={A homotopy method for continuous-time model-free {LQR} control based on policy iteration},
  author={Fan, Wenwu and Xiong, Junlin},
  journal=ieeecaa,
  volume={12},
  number={8},
  pages={1673--1682},
  month={Aug.},
  year={2025},
  publisher={IEEE}
}

@article{chen2023adaptive,
  title={Adaptive optimal control of unknown nonlinear systems via homotopy-based policy iteration},
  author={Chen, Ci and Lewis, Frank L and Xie, Kan and Xie, Shengli},
  journal=ieeeac,
  volume={69},
  number={5},
  pages={3396--3403},
  month={May},
  year={2024},
  publisher={IEEE}
}

@article{ma2025adaptive,
  title={Adaptive dynamic programming for optimal control of unknown {LTI} system via interval excitation},
  author={Ma, Yong-Sheng and Sun, Jian and Xu, Yong and Cui, Shi-Sheng and Wu, Zheng-Guang},
  journal=ieeeac,
  volume={70},
  number={7},
  pages={4896--4903},
  month={Jul.},
  year={2025},
  publisher={IEEE}
}

@article{wu2025memory,
  title={Memory-Efficient Inverse Reinforcement Learning for Multiplayer Differential Games},
  author={Wu, Jiacheng and Zhu, Yang and Su, Hongye},
  journal=ieeec,
  volume={55},
  number={11},
  pages={5545--5558},
  month={Nov.},
  year={2025},
  publisher={IEEE}
}

@article{wang2025parallel,
  title={A Parallel Homotopic Optimized Control Scheme of Uncertain Nonlinear {M}arkov Jump Systems and Its Applications},
  author={Wang, Jing and Huang, Zheng and Shen, Hao and Park, Ju H},
  journal=ieeease,
  volume={22},
  pages={19403--19414},
  year={2025},
  publisher={IEEE}
}

@article{li2025cooperative,
  title={Cooperative Optimal Output Tracking for Discrete-Time Multiagent Systems: Stabilizing Policy Iteration Frameworks},
  author={Li, Dongdong and Dong, Jiuxiang},
  journal=ieeeac,
  volume={71},
  number={4},
  pages={2746--2753},
  month={Apr.},
  year={2026},
  publisher={IEEE}
}

@article{luo2026recent,
  title={Recent advances on off-policy reinforcement learning for optimization control},
  author={Luo, Biao and Liu, Derong and Wu, Huai-Ning and Huang, Tingwen and Yang, Chunhua and Gui, Weihua},
  journal=ieeec,
  year={in press, DOI: 10.1109/TCYB.2026.3683384 },
  publisher={IEEE}
}

@article{lian2022data,
  title={Data-driven inverse reinforcement learning control for linear multiplayer games},
  author={Lian, Bosen and Donge, Vrushabh S and Lewis, Frank L and Chai, Tianyou and Davoudi, Ali},
  journal=ieeennls,
  volume={35},
  number={2},
  pages={2028--2041},
  month={Feb.},
  year={2024},
  publisher={IEEE}
}

@book{hull2013optimal,
  title={Optimal {C}ontrol {T}heory for {A}pplications},
  author={Hull, David G},
  year={Cham, Switzerland: Springer, 2003.},
}

@article{jiang2025off,
  title={Off-Policy Reinforcement Learning for {$H_{\infty}$} Control of Linear Discrete-Time Systems with Network Induced Dropouts},
  author={Jiang, Yi and Yang, Tao and Gao, Weinan and Wu, Jin and Chai, Tianyou and Lewis, Frank L},
  journal=ieeeac,
  volume={70},
  number={12},
  pages={8000--8015},
  month={Dec.},
  year={2025},
  publisher={IEEE}
}

@article{wu2025distributed,
  title={Distributed {FilterNet} reinforcement learning for achieving output consensus in heterogeneous multiplayer multiagent systems},
  author={Wu, Jiacheng and Lian, Bosen and Wen, Changyun and Zhu, Yang},
  journal=ieeennls,
  volume={37},
  number={2},
  pages={575--588},
  month={Feb.},
  year={2026},
  publisher={IEEE}
}

@article{wang2026secure,
  title={Secure optimal control of {It{\^o}} stochastic {M}arkov jump systems subject to {DoS} attacks: A hybrid learning algorithm},
  author={Wang, Xin and Kong, Linghuan and Zhu, Quanxin and Niu, Ben},
  journal=auto,
  volume={183},
  pages={112681},
  year={2026},
  publisher={Elsevier}
}

@article{zhang2025prescribed,
  title={Prescribed-time observer-based {HI-RL} secure output tracking control for heterogeneous {MASs} under {DoS} attacks},
  author={Zhang, Shuo-Qiu and Che, Wei-Wei and Wu, Zheng-Guang},
  journal=ieeesmcs,
  volume={56},
  number={1},
  pages={709--723},
  month={Jan.},
  year={2026},
  publisher={IEEE}
}

@article{qasem2023experimental,
  title={Experimental validation of data-driven adaptive optimal control for continuous-time systems via hybrid iteration: {A}n application to rotary inverted pendulum},
  author={Qasem, Omar and Gutierrez, Hector and Gao, Weinan},
  journal=ieeeie,
  volume={71},
  number={6},
  pages={6210--6220},
  month={Jun.},
  year={2024},
  publisher={IEEE}
}

@article{rizvi2018output,
  title={Output feedback {Q}-learning for discrete-time linear zero-sum games with application to the {H}-infinity control},
  author={Rizvi, Syed Ali Asad and Lin, Zongli},
  journal=auto,
  volume={95},
  pages={213--221},
  year={2018},
  publisher={Elsevier}
}

@article{walsh2002stability,
  title={Stability analysis of networked control systems},
  author={Walsh, Gregory C and Ye, Hong and Bushnell, Linda G},
  journal=ieeecst,
  volume={10},
  number={3},
  pages={438--446},
  month={May},
  year={2002},
  publisher={IEEE}
}

@article{lopez2023efficient,
  title={Efficient off-policy {Q}-learning for data-based discrete-time {LQR} problems},
  author={Lopez, Victor G and Alsalti, Mohammad and M{\"u}ller, Matthias A},
  journal=ieeeac,
  volume={68},
  number={5},
  pages={2922--2933},
  month={May},
  year={2023},
  publisher={IEEE}
}

\end{document}